\documentclass[pdflatex,sn-mathphys-num]{sn-jnl}
 
\usepackage{lmodern}
\usepackage{graphicx}%
\usepackage{multirow}%
\usepackage{amsmath,amssymb,amsfonts}%
\usepackage{amsthm}%
\usepackage{mathrsfs}%
\usepackage[title]{appendix}%
\usepackage[dvipsnames]{xcolor}%
\usepackage{textcomp}%
\usepackage{manyfoot}%
\usepackage{booktabs}%
\usepackage{algorithm}%
\usepackage{algorithmicx}%
\usepackage{algpseudocode}%
\usepackage{listings}%
\usepackage{bm}

\usepackage{anyfontsize}
\usepackage{color}
 \usepackage{diagbox}
\usepackage{enumitem}
\usepackage{tablefootnote}
 \usepackage{amsmath}
 \usepackage{tikz}

\DeclareRobustCommand{\VAN}[3]{#2}

\newtheorem{theorem}{Theorem}

\newtheorem{example}{Example}%
\newtheorem{remark}{Remark}%
\newtheorem{definition}{Definition}%
 \newtheorem{fact}{Fact}
 \newtheorem{lemma}{Lemma}%
  \newtheorem{corollary}{Corollary}%

\newcommand{\lr}[1]{\langle #1 \rangle}

\newcommand{\DEL}{\mathsf{DEL}}
\newcommand{\X}{\mathsf{X}}
\newcommand{\Y}{\mathsf{Y}}
\newcommand{\T}{\mathcal{T}}
\newcommand{\bR}{\mathbf{R}}

\newcommand{\pt}{P{\bm t}}

\newcommand{\bP}{\ensuremath{\mathbf{P}}}

\newcommand{\bI}{\mathbf{I}}
\newcommand{\bD}{\mathbf{D}}

\newcommand{\term}{\mathsf{Term}}
\newcommand{\ELCR}{\mathsf{ELCR}}
\renewcommand{\L}{\mathcal{L}}
\newcommand{\LB}{\L_{\mathsf{B}}}
\newcommand{\weg}[1]{}

\allowdisplaybreaks[4]

\tikzset{global scale/.style={
    scale=#1,
    every node/.append style={scale=#1}
  }
}

\begin{document}

\title{
 Reasoning under uncertainty in the game of Cops and Robbers}

 \author[1,2]{\fnm{Dazhu} \sur{Li}}

\author[3]{\fnm{Sujata} \sur{Ghosh}}

\author[4,5]{\fnm{Fenrong} \sur{Liu}}

\affil[1]{\orgdiv{Institute of Philosophy}, \orgname{Chinese Academy of Sciences}, \orgaddress{\city{Beijing},
\country{China}}}

\affil[2]{\orgdiv{Department of Philosophy}, \orgname{University of Chinese Academy of Sciences}, \orgaddress{\city{Beijing},
\country{China}}}

\affil[3]{\orgdiv{Indian Statistical Institute}, \orgaddress{\city{Chennai},
\country{India}}}

\affil[4]{\orgdiv{The Tsinghua-UvA JRC for Logic, Department of Philosophy}, \orgname{Tsinghua University}, \orgaddress{\city{Beijing},
\country{China}}}

\affil[5]{\orgdiv{Institute for Logic, Language and Computation}, \orgname{University of Amsterdam}, \orgaddress{\city{Amsterdam},
\country{The Netherlands}}}


\abstract{
The game of  Cops and  Robbers is an important model for studying  computational queries in pursuit-evasion environments, among others. As recent logical explorations have shown, its structure exhibits appealing analogies with modal logic. In this paper, we enrich the game with a setting in which players may have imperfect information. We propose a new formal framework,  Epistemic Logic of Cops and Robbers $(\ELCR)$, to make the core notions of the game precise, for instance, players' positions, observational power and inference. Applying $\ELCR$ to analyze the game, we obtain an automated way to track interactions between players and characterize their information updates during the game. The update mechanism is defined by a novel dynamic operator, and we compare it with some relevant paradigms from the game and logic perspectives. We study various properties of  
$\ELCR$ including axiomatization and decidability. To our knowledge, this is the first attempt to explore these games from a formal point of view where (partial) information available to players is taken into account.}

\keywords{Cops and Robbers, Hide and Seek, Imperfect information, Observation power, Knowledge updates, Dynamic-epistemic logic, Axiomatization, Decidability}



\maketitle

\section{Introduction}\label{sec:introduction}
 
 Search missions and pursuit-evasion environments have been investigated in details in the study of robotic systems. Such pursuit-evasion problems, which can also be viewed as adversarial search problems, can be considered as strategizing problems in pursuit-evasion games or their  multi-agent counterpart, {\em Cops and Robbers}. The game of Cops and Robbers refers to a family of pursuit-evasion games played on graphs, where the pursuer  (or the cops) must develop optimal responses or search strategies against some worst-case adversary, the evader (or the robbers).
The game has a deep root in computer science and serves as an important testbed for studying algorithms and computational complexity  in search environments \cite{cop-robber,cop-robber-book}. 

In recent years, logicians have studied the Hide and Seek game, which shares a similar pursuit-evasion structure, with a focus on modeling the dynamic interaction between players  \cite{original-sg,graphgame}. A two-dimensional modal logic for the game is proposed in \cite{graphgame}, and its further properties and extensions are explored in  \cite{lhs, LHS-journal, Chenqian2023, hlhs-axiomatization}. However, these logical models assume that players have perfect information, which limits the resulting frameworks to tools for reasoning solely about properties of graph structures. In this paper, we shift to an imperfect information setting, examining how agents reason and play the game under uncertainty due to their limited sight or observational power.  So, the game of Cops and Robbers studied in the paper can be seen as an imperfect information version of Hide and Seek explored in \cite{lhs, LHS-journal, Chenqian2023, hlhs-axiomatization}. Our focus is on how they use partial information available to them during play. We propose a formal framework to support the study of these interactions. 
 From a practical viewpoint,  graph games have been used extensively to model reachability problems, social networks, and search problems. Our framework provides a formal language to precisely describe and study a wide range of such issues. This formal study can facilitate the autonomous agent-building efforts with a better know-how in terms of information available to these agents. The corresponding information updates  during gameplay help address adversarial search problems. 

To illustrate the game’s features, let us first consider a specific example.
For simplicity, the game considered in this context is played by two players, a Cop and a Robber, on a {\em finite directed} graph in which every vertex has successors, an assumption made for ease of presentation.\footnote{The game and logic introduced in this paper can be generalized to a setting with more players.}

\vspace{1mm}

\begin{example}\label{ex1}
In the graph below, Cop $\X$ (female) is at $0$, and Robber $\Y$ (male) is at $4$. They know the graph structure and their own positions. A player can see the other if they are at the same position or at a vertex reachable by an arrow in either direction.
Thus, $\X$ and $\Y$ do not know each other's exact positions. However, given the graph structure and their observational range, $\X$ knows that $\Y$ must be at $2$, $3$, or $4$, while $\Y$ knows that $\X$ must be at $0$ or $1$.

  
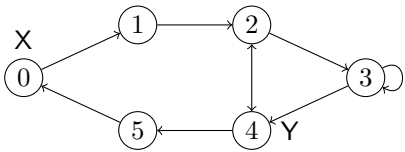
\vspace{1mm}
  
\begin{center}
\begin{tikzpicture}
\node(0)[circle,draw,inner sep=0pt,minimum size=5mm] at (0,0)[label=above:$\X$] {$0$};
\node(1)[circle,draw,inner sep=0pt,minimum size=5mm] at (1.5,.7) {$1$};
\node(2)[circle,draw,inner sep=0pt,minimum size=5mm] at (3,.7) {$2$};
\node(3) [circle,draw,inner sep=0pt,minimum size=5mm] at (4.5,0){$3$};
\node(4) [circle,draw,inner sep=0pt,minimum size=5mm] at (3,-.7)[label=right:$\Y$]{$4$};
\node(5) [circle,draw,inner sep=0pt,minimum size=5mm] at (1.5,-.7){$5$};
 
\draw[->] (0) to  (1);
\draw[->] (1) to  (2);
\draw[->](2) to  (3);
\draw[->](3) to  (4);
\draw[->](4) to  (5);
\draw[->](5) to  (0);
\draw[<->](2) to  (4);
\draw[->](3) to  [in=330, out=30,looseness=4] (3);
\end{tikzpicture}
\end{center}

  \vspace{1mm}
  
\noindent A player whose turn it is has to move along an arrow. Let  $\X$ act first. She can only move to $1$, and after this, she knows that $\Y$ is not at $2$, otherwise, she would see him directly. So, $\X$ knows that $\Y$ must be at $3$ or $4$. For $\Y$, he knows after her move that $\X$ is not at $2$ (as  otherwise he would see her directly) and not at $0$ (since none of his previously considered positions for $\X$—$0$ or $1$—can reach $0$ in one step). Therefore, $\Y$ concludes that $\X$ must now be at $1$.\footnote{At this stage $\Y$ also knows that $\X$ was at $0$ before the movement. Although $\Y$ knows the actual action of $\X$, the knowledge of $\Y$ is obtained by his reasoning with the knowledge about graph and his observational power. There is no inconsistency with the imperfect information nature of the game: for instance, if we let the game begin with the movement of $\Y$ from $4$ to $2$, then after the movement $\X$ would not know the actual action of $\Y$, since $\X$ would think it is  also  possible that, e.g., $\Y$ moves from $3$ to $3$.} 


Next, let $\Y$ move to $5$. Since $\X$ observes that $\Y$ is no longer at $2$, she knows that $\Y$ must be at $3$, $4$, or $5$. $\X$ then moves to $2$, from where she can see that $\Y$ is not at $3$ or $4$. Thus, $\X$ concludes that $\Y$ must be at $5$ and wins.
\end{example}

\vspace{1mm}

Although the example is simple, it indicates several subtleties of the game. For instance, the players' knowledge changes continuously throughout the game, and  in the final stage, although the players cannot see each other directly, they can know the positions of the other based on their knowledge about the graph structure and their observational power.
 Knowledge here is modelled based on observational powers and expresses the uncertainty of the players. In essence, it is quite similar to the way it is modelled in imperfect-information games in general (see e.g., \cite{gamesinformation}), where information is expressed in terms of equivalence relations, but ultimately models the uncertainties of the players. 

In what follows, we will introduce the key notions precisely and present an {\em Epistemic Logic of Cops and Robbers} ($\ELCR$), which enables us to reason about the action-information interplay and  to automatically track knowledge changes during play. We are aware of the well-developed methodology of dynamic-epistemic logic ($\DEL$, \cite{BMS,pal,DEL-book,DEL-johan-book}) and its various applications in modelling information update. We will compare the two approaches and  provide examples to show that our framework may offer a more succinct representation. The interface between logic and games has a long history. Over the past two decades, inspired by board games, graph game logic has emerged as an important research area, \cite{sabotagelori,sabotage,dazhu-jolli,original-sg,graphgame,poisonlogic,Chenqian2023,poisonargumentation,Dazhu-thesis,lhs,LHS-journal,hlhs-axiomatization,Declan}, with various logical techniques developed to study player interactions on graphs. To our knowledge, this is the first attempt to explore more complex settings where agents have imperfect information. We believe that the ideas and techniques introduced in this article will help open a new line of research.

Before we begin, it is worth noting that this work extends the previous proceedings version \cite{elcr}. Specifically, compared to the earlier version, this paper offers the following new contributions. 
The logical language is extended to allow dynamic operators within the scope of knowledge operators, and the models are redesigned to better reflect the game’s structure. 
We establish complete calculi for both $\ELCR$ and its  static fragment $\ELCR^-$ (without dynamic operators), further prove that $\ELCR$ is decidable. 
More new examples are included to highlight subtleties of the game and its logic, along with additional properties of the logic to deepen our understanding of the game.
Finally, we provide an extensive discussion of related works, including a comparison of $\ELCR$ with the $\DEL$-approach to the game proposed in a recent paper \cite{johan-del-game}.

\vspace{1mm}

\noindent {\it Structure of the paper}.\; In Section \ref{sec:game}, we lay out the basics of the  game and discuss its alternative designs. Section\,\ref{sec:LCR} presents the formal language and models of $\ELCR$, proposes an alternative for a simultaneous-move variant, and compares the two logics.
Section \ref{sec:application}  applies the new framework to analyze the game of Cops and Robbers.
Section \ref{sec:logical-properties} studies some basic properties of $\ELCR$. Sections \ref{sec:axiomatization-static} and \ref{sec:axiomatization-whole-logic} provide complete Hilbert-style proof systems for the static fragment of the logic without dynamic operators and the whole $\ELCR$ respectively and show that $\ELCR$ has a decidable satisfiability problem.  Section\,\ref{sec:related} formalizes some recent ideas on the $\DEL$-approach to studying the game, compares it to our logic, and discusses relevant works. Section\,\ref{sec:conclusion} concludes with directions for future research.

\section{Game design}\label{sec:game}

Let us first identify basic assumptions regarding players' knowledge. First, given a game, we assume that the graph structure is commonly known by the players. Also, the game is turn-based, and  it is common knowledge that whose turn it is to move: here we simply require that in each round, Cop $\X$ moves first, and then  Robber $\Y$ moves.\footnote{The turn-based assumption is in line with many other works on the game, e.g., \cite{cop-robber,cop-robber-book}. One can also consider any other specific order of their play, and even simultaneous play, but that will not affect our basic idea to analyze the game and the design of the logical tool in Section \ref{sec:LCR}.} However, they may not know where the opponent moves from and/or where the opponent moves to. Finally, we assume that  players can at least remember what they have considered to be possible at the previous stage: when a player moves,  the players infer the new possible situations from what they considered to be possible before that movement. For instance, before the movement of $\X$ in Example\,\ref{ex1}, $\Y$ considers it possible for $\X$ to be at $0$ or $1$, and once $\X$ moves, he gets to know that $\X$ is at $1$, which is the successor of the previous possibility $0$.

\vspace{1mm}

As for the winning condition, one option is to stipulate that Cop wins iff she is at same position as Robber \cite{lhs,LHS-journal}. However, our exploration on the role of knowledge allows more alternatives to define the condition, and we will adopt a new  criterion:

  \vspace{1mm}

\noindent {\it Winning condition:}\; Fix a natural number $n\in\mathbb{N}$. Cop $\X$ wins iff the position of Robber $\Y$ is {\it known} by $\X$ within $n$ rounds.

  \vspace{1mm}
  
\noindent  {\it Strategy:}\; A strategy of a player is a function from the set of positions of the player to the possible moves at that position. A strategy is said to be winning for a player if the player can win the game by playing according to the strategy, whatever be the moves of the other player.\footnote{For instance, when we assume that the initial positions of $\X$ and $\Y$ in Example \ref{ex1} are $3$ and $5$ respectively, a winning strategy of $\X$ is to always stay at $3$: one can check that $\X$ can know the position of $\Y$ after the movement of $\Y$ from $0$ to $1$ in the second round.}

  \vspace{1mm}
  
\noindent {\it Ability of players: $k$-sight.}\;  To analyze the imperfection information game, it is crucial to identify what players can know, which is dependent on the observational ability of the player.
There are two extremes about their ability:  players are assumed to know the positions of each other at any stage of a game,
and  they do not have any ability to ensure they can know something, even their own positions. 
Inspired by \cite{diffusion,Davide-short-sight,LLS}, 
we introduce the notion of {\em $k$-sight of players} to describe their observational power: players always know  their own positions, and if their positions are reachable within $k$ steps via the arrows or  their converse directions (in which case, we say that {\em they can see each other}), then they know the positions of each other. For instance, $\X$ and $\Y$ in Example \ref{ex1} have sight $1$. Below is another example.

\vspace{1mm}

\begin{example}
Assume that $\X$ and $\Y$ have sight 2. Based on our definition of the $k$-sight ability, to identify what the players can see, the direction of the arrows in the graph below does not matter. 
\begin{center}
\begin{tabular}{ll}
\begin{minipage}[l]{0.6\textwidth}
All of the vertices $s_0$-$s_4$ and $s_6$ are in the sight of $\Y$ from $s_6$,  so $\Y$ knows that $\X$ is not at any of those vertices (which together with the knowledge about the graph structure would make $\Y$ know that $\X$ is at $s_5$). From the vertex $s_5$, $\X$ can see the vertices $s_1$-$s_5$, and based on the knowledge about the graph structure, $\X$ knows $\Y$ is at either $s_0$ or $s_6$, although she does not know which one is exactly the case.  But after $\X$ moves from $s_5$ to $s_2$, the position of $\Y$ would come in the sight of  $\X$, meaning that after the move $\X$ would know  where $\Y$ is.\addtocounter{footnote}{1}\footnotemark[\value{footnote}]
\end{minipage}  & \qquad  \begin{minipage}[r]{0.3\textwidth}
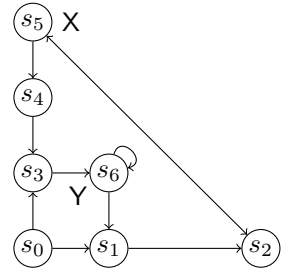
\begin{tikzpicture}
\node(a)[circle,draw,inner sep=0pt,minimum size=5mm] at (0,0)  {$s_0$};
\node(b) [circle,draw,inner sep=0pt,minimum size=5mm] at (1,0){$s_1$};
\node(c) [circle,draw,inner sep=0pt,minimum size=5mm] at (3,0){$s_2$};
\node(d) [circle,draw,inner sep=0pt,minimum size=5mm] at (0,1){$s_3$};
\node(e) [circle,draw,inner sep=0pt,minimum size=5mm] at (0,2){$s_4$};
\node(f) [circle,draw,inner sep=0pt,minimum size=5mm] at (0,3) [label=right:$\X$]{$s_5$};
\node(g) [circle,draw,inner sep=0pt,minimum size=5mm] at (1,1){$s_6$};
\node(h)   at (.6,.7)  {$\Y$};
\draw[->](a) to  (b);
\draw[->](b) to  (c);
\draw[->](a) to  (d);
\draw[->](e) to  (d);
\draw[->](f) to  (e);
\draw[->](g) to  (b);
\draw[->](d) to  (g);
\draw[<->](c) to  (f);
\draw[->](g) to  [in=15, out=75,looseness=3] (g);
\end{tikzpicture}
 \end{minipage}  
\end{tabular}
\end{center}
\end{example}

 \footnotetext[\value{footnote}]{If we temporally let $\Y$ move first in this example, then $\X$ would know where  $\Y$ is after the first movement of $\Y$. The reason is as follows. Player  $\Y$ has two options: moving to $s_1$ or staying at $s_6$. If he moves to $s_1$, then $\X$ can see him directly. When the latter is the case,  $\X$ still cannot see $\Y$ directly, but the movement makes $\X$ know that it is impossible that $\Y$ is at another unobservable vertex $s_0$, since it has no predecessor from the previous possibilities $s_0$ and $s_6$ considered by $\X$, i.e., it would not happen that $\Y$ moves to $s_0$ from $s_0$ or $s_6$. In the latter case,  {\em physically nothing changes, but epistemically the player knows more}.}

 \vspace{2mm}

It is important to emphasize that the $k$-sight ability is  used to characterize what {\em at least} players can know, but depending on the concrete situations players may know more. For instance, at the final stage of the game in  Example \ref{ex1},  although $\X$ cannot observe $\Y$ directly based on the $1$-sight ability,  she can still know where $\Y$ is.  We emphasize that our goal in this work is to develop logical tools to track changes in the players’ knowledge of each other’s positions. In terms of knowledge, as a first step we will only consider the  knowledge of atomic facts, their Boolean combinations, and the effects of movements on them, but not higher-order knowledge (e.g., $\X$ knows that $\Y$ knows that $\X$ is at the vertex $a$). Although this may look restricted, the proposal developed in this way fits many existing analyses for the imperfect information variants of Cops and Robbers (see e.g., \cite{cop-robber-book}). We leave the work for the more intricate setting involving higher-order knowledge to another occasion.

We end this section by pointing out possible options regarding our assumptions of the game. For instance,  as stated earlier, graphs in this context are {\em serial}, but it is equally  reasonable to consider graphs without any restrictions; the players can act simultaneously;  different players may have different sights; there can be more players with other kinds of ability, say, they can send each other messages; and there may be other ways to define the winning condition, for instance, when there are $i > 1$ Cops, Cops may win when the position of Robber is their {\em distributed knowledge}.  Finally, it is also meaningful to extend our current game with explicit probabilities of actions in different situations, a usual manner adopted in imperfect information games, which would affect how players update their knowledge and how they act.\footnote{Although we do not use probabilities explicitly, how players update their knowledge in our setting in effect is involved with probabilities in an implicit way, which are determined by, e.g., how many positions of a player are considered to be possible by the other, how many successors those possible positions have and what a player can see directly. In line with the implicitness, we will propose a  qualitative approach to reason about the current game in Section \ref{sec:LCR}.} Some of these alternatives will be discussed in the later sections. For now, let us move to the details of the logical framework.

\section{Logical language and models}\label{sec:LCR}
This section will present a logical framework that characterizes our assumptions about the game and enables us to reason about  how players update their knowledge.  

Inspired by  \cite{knowing-value,lfd}, we will use different {\em values} to encode different vertices in the graph of a game. Such values and the binary relation in a graph give us a semantic structure of first-order logic $(\mathsf{FOL})$. Also, we will use {\em variables} to denote players, and then the current position of a player gives us the value of the corresponding variable. So, positions of all players can give us an {\em assignment} function $\sigma$ that assigns values to variables, say, $\sigma(x) = a$ when player $x$ is at $a$.

 \vspace{1mm}

\begin{remark}
The idea above to define the logic seems  similar to $\mathsf{FOL}$, but there is a crucial difference concerning the usages of variables. In $\mathsf{FOL}$, variables are just placeholders without any intrinsic meaning, while variables in our proposal denote {\em positions of players} and  can take different values in different situations of a game.  This use of variables aligns with that in some recent dependence logics (e.g., \cite{lfd,dependence-logic-book}) and in, e.g., physics (for instance,  usually we use ``$v$" for velocity).
\end{remark}

 \vspace{1mm}

 We fix a {\em vocabulary} $Voc = (Pred, Cons, Var)$, where $Pred$ is a set of predicate symbols, containing a specific binary relation symbol $R$ describing the arrows of a game graph, $Cons$ is a non-empty, finite set of constants, and $Var=\{x,y\}$, meaning the players.  As usual, elements of $\term=Cons\,\cup\, Var$ are called {\em terms}. The language $\mathcal{L}$ for {\em the Epistemic Logic of Cops and Robbers $(\ELCR)$} is defined in the following:

 \vspace{1mm}

\begin{definition}\label{def:language}
 Formulas in {\em the language $\mathcal{L}$ for $\ELCR$} are defined as follows:
 
    \vspace{1mm}
 \begin{center}
$\mathcal{L}_{\mathsf{B}}\ni\alpha::=\pt\mid t_1\equiv t_2\mid \neg \alpha\mid  (\alpha\land\alpha)$ \vspace{1mm} \\ 
$\mathcal{L}_{\mathsf{BD}}\ni\psi::=\alpha\mid \neg \psi\mid  (\psi\land\psi)\mid [z]\psi$ \vspace{1mm} \\ 
$\mathcal{L}\ni\varphi::= \psi \mid  K_zt \mid  \neg \varphi\mid  (\varphi\land\varphi)\mid K_z\psi\mid [z]\varphi$
 \end{center}

   \vspace{1mm}
  
\noindent where $t,t_1,t_2\in Cons\cup Var$ are terms,  ${\bm t}$ is a tuple of terms, $P\in Pred$ is a predicate symbol, and $z\in Var=\{x,y\}$ is a variable.  Other Boolean connectives $\top, \bot, \lor, \to, \leftrightarrow$ are defined as usual. We use $\lr{K_z}\varphi$ and $\lr{z}\varphi$ for $\neg K_z\neg\varphi$ and $\neg[z]\neg\varphi$, respectively. Also, for a set $\T\subseteq Cons\cup Var$, we use  $K_z\T$ for $\bigwedge_{t\in \T}K_zt$.
\end{definition}

\vspace{1mm}

 In the definition, one can merge $\mathcal{L}_{\mathsf{BD}}$ and $\mathcal{L}_{\mathsf{B}}$ by adding the dynamic operators to $\mathcal{L}_{\mathsf{B}}$ directly, but to make it easier to reference the static Boolean formulas, we stick to the current form of the definition. We write $\L^{-}$ for {\em the fragment of $\L$ without dynamic operators} and $\ELCR^{-}$ for the corresponding logic. As usual, $\pt$ and $t_1\equiv t_2$ are {\em atomic formulas} (in particular, formula $Rt_1t_2$ means  {\em the value of $t_2$ is a successor of the value of $t_1$}),  $K_zt$ reads {\em player $z$ knows the value of $t$}, $K_z\varphi$ means {\em $z$ knows that $\varphi$}, and $[z]\varphi$ expresses {\em after any movement of $z$, $\varphi$ is the case.}  The language  does not contain formulas for higher-order knowledge (e.g., $K_xK_yc$). To define both the changes of knowledge for positions and the winning positions  in the game, we even do not need the formulas of the form $K_z\psi$,  but we add them to the language for convenience.

 \vspace{1mm}

\begin{definition}\label{def:model-predicate}
 A {\em  model for $\ELCR$} is a tuple $M=(\bD,\bI, \Sigma,\sim)$, where

    \vspace{1mm}
   
\begin{itemize}
\item[$\bullet$] $\bD$ is a  non-empty, finite set of values  (also called vertices or positions).

  \vspace{1mm}
\item[$\bullet$]  $\bI$ is the interpretation function such that 
  \vspace{1mm}
\begin{itemize}
    \item[$\bullet$] For each $m$-ary predicate symbol $P\in Pred$,  $\bI(P)\subseteq \bD^m$ is an $m$-ary relation on $\bD$. In particular, $\bI(R)$ is a binary relation on $\bD$ such that for any $s\in\bD$, there is some $t\in\bD$ such that $(s,t)\in\bI(R)$.
      \vspace{1mm}
    \item[$\bullet$] For each $c\in Cons$, $\bI(c)\in\bD$. Moreover, for each $s\in\bD$, there is $c\in Cons$ such that $\bI(c)=s$.
\end{itemize}  
  \vspace{1mm}
\item[$\bullet$] $\Sigma\subseteq \bD^{Var}$ is a  non-empty set of {\em situations} of the players' positions, which are also called {\em assignments}.  
  \vspace{1mm}
\item[$\bullet$] For each $z\in\{x,y\}$, $\sim_z\subseteq \Sigma\times \Sigma$ is an equivalence relation. 
\end{itemize}
\end{definition}

 \vspace{1mm}

When the interpretation function $\bI$ is clear, we often write $\bP$ for $\bI(P)$. Intuitively, $(\bD,\bR)$  represents a graph where a game is played, and by definition, $\bR$ is serial; $\Sigma$ is a collection of possible situations of a game; and the relations $\sim_{z\in Var}$ are the indistinguishability relations of players  
(e.g., $\sigma_1\sim_x\sigma_2$ means that player $x$ cannot distinguish between situations $\sigma_1$ and $\sigma_2$).

The class of all  models captures the case where players do not have any ability, not even the $0$-sight: for instance, it might be the case that $\sigma_1\sim_x\sigma_2$ and $\sigma_1(x)\not=\sigma_2(x)$, which means that $x$ does not know where herself is. To capture the $k$-sight ability, we first define an auxiliary notion as follows: 

 \vspace{1mm}

\begin{definition}\label{def:distance}
Let $(\bD,\bR)$ be a finite graph. For any $s\in \bD$ and $m\in \mathbb{N}$, we inductively define the following: 
 \begin{center}
    $\mathbb{D}^{0}(s):=\{s\}$\vspace{1mm} \\ $\mathbb{D}^{m+1}(s):=\mathbb{D}^{m}(s)\cup\{t\in \bD: \exists u\in \mathbb{D}^{m}(s) \;\textit{s.t.}\; (u,t)\in \bR \;  \textit{or}\;  (t,u)\in \bR\}.$
\end{center}
\end{definition}

 \vspace{1mm}

So, $t\in \mathbb{D}^k(s)$  states that the `distance' between $s$ and $t$ is not more than $k$, and more precisely,  $t$ can be reached from $s$ within $k$ steps via the symmetric closure of $\bR$. 

  \vspace{1mm}

\begin{definition}\label{def:k-sight-predicate-model} A
{\em $k$-sight model} is a  model $M=(\bD,\bI, \Sigma,\sim)$  such that for each $z\in Var$ and $\sigma,\sigma'\in \Sigma$, if $\sigma\sim_z\sigma'$, then for any $z'\in Var$ with $\sigma(z')\in \mathbb{D}^k(\sigma(z))$,   $\sigma(z')=\sigma'(z')$.
\end{definition}

 \vspace{1mm}

In a $k$-sight model, if the distance between the position of a player $z$ and the position of player $z'$ is not more than $k$, then in all situations that cannot be distinguished by $z$, player $z'$ always has the same position, which means that $z$ knows the position of $z'$. This  characterizes the $k$-sight ability. In what follows, we work with $k$-sight models, and $\ELCR$ is defined based on them.

  \vspace{1mm}

\begin{remark}\label{remark:sight} 
The notion of $k$-sight models can be adapted to capture cases that players have different sights. For each $z\in\{x,y\}$, we use $k_z$ for the sight of the player. Then, to characterize this more complicated setting, we just need to replace the restriction imposed on $k$-sight models with the following: for any  $z\in Var$ and  $\sigma,\sigma'\in \Sigma$,

 \vspace{1mm}

\begin{center}
  if $\sigma\sim_z\sigma'$, then for any $z'\in Var$ with $\sigma(z')\in \mathbb{D}^{k_z}(\sigma(z))$, $\sigma(z')=\sigma'(z')$.  
\end{center}
\end{remark}

 \vspace{1mm}

Given a model $M=(\bD,\bI, \Sigma,\sim)$ and $\sigma\in\Sigma$, for any  $t\in\term$, we use $t^{(\bI,\sigma)}$ for the value of $t$: if $t\in Cons$, then $t^{(\bI,\sigma)}:=\bI(t)$; and if $t\in Var$, then $t^{(\bI,\sigma)}:=\sigma(t)$.

Now we move to presenting the semantics for the logic.  Truth conditions for the fragment  $\mathcal{L}^-$ are straightforward, and key clauses are as follows:

   \vspace{1mm}
 \begin{center}
\begin{tabular}{rcl}
$M,\sigma\models P(t_1,\dots,t_n)$     & \quad  iff  & \quad $(t_1^{(\bI,\sigma)},\dots, t_n^{(\bI,\sigma)})\in \bI(P)$\\
$M,\sigma\models t_1\equiv t_2$  & \quad  iff  & \quad $t_1^{(\bI,\sigma)}=t_2^{(\bI,\sigma)}$\\
 $M,\sigma\models K_zt$     & \quad  iff  & \quad   for all $\sigma'\in \Sigma$, if $\sigma'\sim_z \sigma$, then $t^{(\bI,\sigma)}= t^{(\bI,\sigma')}$\\
 $M,\sigma\models K_z\varphi$    & \quad  iff  & \quad   for all $\sigma'\in \Sigma$, if $\sigma'\sim_z \sigma$, then  $M,\sigma'\models\varphi$\\
\end{tabular}
\end{center}

   \vspace{1mm}

Recall that in our models $M$, for each vertex, there is a constant true at the vertex. So, given a situation $\sigma$, for any term $t\in\term$, there is a constant $c$ such that $t\equiv c$ at $\sigma$ (if $t$ is a constant, then $t\equiv c$ is globally true in $\Sigma$). Based on this, $K_zt$ amounts to $\bigwedge_{c\in Cons}(\lr{K_z}t\equiv c\to K_zt\equiv c)$, but we stick to using $K_zt$ for syntactic succinctness.

It remains to show the truth condition for $[z]\varphi$. In terms of the game, a desired clause would give us an automatic mechanism to capture the effects of movements, which are involved with both the changes of positions and the updates of the knowledge. Let us first define binary relations $\mathsf{R}^{z\in Var}$ on {\em all} assignments $\bD^{Var}$:

 \vspace{1mm}

\begin{center}
For any $\sigma,\sigma'\in \bD^{Var}$, we write $\mathsf{R}^z\sigma\sigma'$ if $(\sigma(z),\sigma'(z))\in \bR$ and for $z'\in Var\setminus\{z\}$, $\sigma(z')=\sigma'(z')$.    
\end{center}

   \vspace{1mm}
  
\noindent Therefore, when $\mathsf{R}^z\sigma\sigma'$, the only difference  of $\sigma$ and $\sigma'$ concerns the values of $z$. It intuitively describes the fact that after $z$ moves from $\sigma(z)$ to $\sigma'(z)$, the situation $\sigma$ becomes $\sigma'$. Given a class   $\Sigma\subseteq  \bD^{Var}$ of situations and a variable $z$, we define 
 \vspace{1mm}
\begin{center}
$\mathsf{R}^z(\Sigma):=\{\sigma'\mid \textit{there is~}\sigma\in\Sigma \textit{~s.t.~}\mathsf{R}^z\sigma\sigma'\}$.
\end{center}
 \vspace{1mm}
\noindent When $\Sigma$ is a singleton $\{\sigma\}$, we write $\mathsf{R}^z(\sigma)$ for $\mathsf{R}^z(\{\sigma\})$. Also, we use $\Sigma|\sigma$ to mean the subset $\Sigma'=\{\sigma'\in\Sigma\mid \sigma \sim_z \sigma' \; \textit{for some}\; z\in Var\}$. Now we can show the truth condition for the movement operators $[z]\varphi$:\footnote{The authors would like to thank Alexandru Baltag for a  useful discussion on the notion of the update.}

\vspace{1mm}

\begin{definition}\label{def:updates}
For the case that both the players have the same sight $k$, {\em the truth condition for $[z]\varphi$} is as follows:

   \vspace{1mm}
  
\begin{center}
$(\bD, \bI, \Sigma, \sim), \sigma_1\models [z]\varphi$\quad iff \quad  for all $\sigma_2\in \mathsf{R}^z(\sigma_1)$, $(\bD, \bI, \Sigma', \sim'), \sigma_2\models \varphi$,
\end{center}

   \vspace{1mm}
  
\noindent where the new $\Sigma'$ is given by the following:

   \vspace{1mm}
  
\begin{itemize}
    \item[$(a)$] if $\sigma_2(x)\in \mathbb{D}^k(\sigma_2(y))$, then $\Sigma'=\{\sigma_2\}$,

      \vspace{1mm}
      
    \item[$(b)$] if $\sigma_2(x)\not\in \mathbb{D}^k(\sigma_2(y))$, then $\Sigma'=\{\sigma'\in  \mathsf{R}^z(\Sigma|\sigma_1)\mid \sigma'(x)\not\in \mathbb{D}^k(\sigma'(y))\}$,
\end{itemize}

   \vspace{1mm}
  
\noindent and the new relations $\sim'_{z\in\{x,y\}}$ on $\Sigma'$ are obtained by the following:

   \vspace{1mm}
  
\begin{itemize}
    \item[$(c)$]  $\sigma'_1\sim'_z\sigma'_2$ iff $\sigma'_1(z)=\sigma'_2(z)$. 
\end{itemize}
\end{definition}

 \vspace{1mm}

The clause $(a)$ aims to deal with the case that after the movement, the two players are in the sight of each other. In this case, both of them know the actual situation. Moreover, the clause $(b)$ tackles the case that after the movement, the two players are not in the sight of each other. Different from that of $(a)$, players only consider the situations in which they are not in the sight of each other to be possible. 

With the definition above, one can see that 
different $\sigma_2\in \mathsf{R}^z(\sigma_1)$ can give us different $\Sigma'$. It  always holds that $\Sigma'\subseteq \mathsf{R}^z(\Sigma)$ and $\sigma_2\in\Sigma'$. 
When $M$ is a $k$-sight model, the resulting model obtained by the updating $M$   is again a $k$-sight model. 
We will provide some examples for the updates further in Section \ref{sec:application}.  In what follows, we will use some ordinary notions directly, including {\em satisfiability}, {\em validity} and {\em logical consequence}, which can be defined in the usual manner.

In the remainder of the section, we provide some comments on the update mechanism and offer the method of updates for {\em the simultaneous movements}. For convenience, we will often use tuples of values to specify the positions of $x,y$ (in this order), to denote situations. For instance,   we write  $\sigma$ as $(s_1,s_2)$ if $\sigma(x)=s_1$ and $\sigma(y)=s_2$. Also, we will highlight the actual situation with an underline, e.g., $\underline{(s_1,s_2)}$.

 \vspace{1mm}

\begin{remark}\label{remark:update}
In the clause $(b)$ of Definition \ref{def:updates}, we require $\Sigma'$ is a subset of $\mathsf{R}^{z}(\Sigma|\sigma)$, although it looks natural to just require  it to be a subset of $\mathsf{R}^{z}(\Sigma)$. 
However,  {\em restricting our attention to $\mathsf{R}^{z}(\Sigma|\sigma)$   is vital to ensure that the update would not cause  disorders.}  To see this, let us consider an example where we just require $\Sigma'$ to be a subset of $\mathsf{R}^{z}(\Sigma)$. 

 
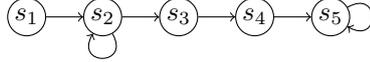
\vspace{1mm}

\begin{center}
\begin{tikzpicture}
\node(a)[circle,draw,inner sep=0pt,minimum size=5mm] at (0,0) {$s_1$};
\node(b) [circle,draw,inner sep=0pt,minimum size=5mm] at (1,0){$s_2$};
\node(c)[circle,draw,inner sep=0pt,minimum size=5mm] at (2,0) {$s_3$};
\node(d)[circle,draw,inner sep=0pt,minimum size=5mm] at (3,0) {$s_4$};
\node(e)[circle,draw,inner sep=0pt,minimum size=5mm] at (4,0) {$s_5$};
\draw[->] (a) to  (b);
\draw[->](b) to  (c);
\draw[->](c) to  (d);
\draw[->](b) to  [in=240, out=300,looseness=5] (b);
\draw[->](d) to  (e);
\draw[->](e) to  [in=330, out=30,looseness=5] (e);
\end{tikzpicture}
\end{center}

We consider  $\Sigma=\{\underline{(s_1,s_4)}, (s_2,s_5)\}$, where $(s_1,s_4)$ 
is the actual situation. Also, we assume that both $x$ and $y$ have sight $1$, but due to some reason they happen to be able to distinguish between the two situations (so both $K_xy$ and $K_yx$ hold at $(s_1,s_4)$). 

Now, let $x$ move, after which the only possible situation is $(s_2,s_4)$. However,
 \vspace{1mm}

\begin{center}
$\Sigma'=\{\sigma'\in \mathsf{R}^x(\Sigma)\mid \sigma'(x)\not\in \mathbb{D}^k(\sigma'(y))\}=\{\underline{(s_2,s_4)}, (s_2,s_5), (s_3,s_5)\}$
\end{center}

 \vspace{1mm}

\noindent At the new situation $(s_2,s_4)$, by our clause for the epistemic relations, we still have $K_yx$, but $x$ now cannot distinguish between $(s_2,s_4)$ and $(s_2,s_5)$, i.e., $x$ forgot the position of $y$ immediately after the movement of $x$ herself! However, this should not be the case in our setting: see Fact \ref{fact:movement-theother}.\footnote{It is important to study the case that players have only {\em restricted memory} \cite{diversity,reaoning-about-knowledge}, 
and we leave this to future inquiry.}

Finally, it is worth noting that the difference between the two requirements just make sense for the initial updates of a given model, in that for the resulting $\Sigma'$ after any update associated to the new situation $\sigma'$, it is always the case that  $\Sigma'|\sigma'=\Sigma'$.
\end{remark}

  \vspace{1mm}

\noindent {\bf Digression: Simultaneous movements.}\; Let us end this part by a brief discussion on the logic for the setting that players $x,y$ move simultaneously.
 To capture the effects of the new pattern, it is crucial to define suitable `group movement operators' $[Var]\varphi$, and as what we are going to show, such a framework can be easily obtained by generalizing the idea behind the updates induced by $[z]$. 
Let us define $\mathsf{R}^{Var}$ such that
 
  \vspace{1mm}
 
 \begin{center}
     for any $\sigma,\sigma'\in\bD^{Var}$, $\mathsf{R}^{Var}\sigma\sigma'$ iff for any $z\in Var$, $(\sigma(z),\sigma'(z))\in \bR$,
 \end{center} 
 
  \vspace{1mm}
 
\noindent which captures the changes of situations caused by the simultaneous movements of both $x$ and $y$. For a set  $\Sigma\subseteq \bD^{Var}$, we can define 

 \vspace{1mm}

\begin{center}
  $\mathsf{R}^{Var}(\Sigma):=\{\sigma'\mid \textit{there is~}\sigma\in\Sigma \textit{~s.t.~}\mathsf{R}^{Var}\sigma\sigma'\}$,  
\end{center}

 \vspace{1mm}

\noindent and again, when $\Sigma$ is a singleton $\{\sigma\}$, we write  $\mathsf{R}^{Var}(\sigma)$ for $\mathsf{R}^{Var}(\{\sigma\})$. Now, the truth condition for $[Var]\varphi$ is given by the following:

 \vspace{1mm}

\begin{center}
    \begin{tabular}{rcl}
    $(\bD, \bI, \Sigma, \sim), \sigma_1\models [Var]\varphi$  & \;\; iff \;\; & for all $\sigma_2\in \mathsf{R}^{Var}(\sigma_1)$, $(\bD, \bI, \Sigma', \sim'), \sigma_2\models \varphi$,  
    \end{tabular}
\end{center}

 \vspace{1mm}

\noindent where  $\Sigma'$  and  $\sim'$ are given by the same clauses as that in Definition \ref{def:updates}, except we now need to use $\mathsf{R}^{Var}$ instead of $\mathsf{R}^{z}$. 
The logic for the simultaneous movements is interesting in its own right, as suggested by the following:

 \vspace{1mm}

\begin{fact}\label{fact:validity}
  Let $\L^{+}$ be the language extending $\L$ with $[Var]\varphi$. The following  is not valid:
  
   \vspace{1mm}
  
  \begin{center}
      $[x][y]\varphi\leftrightarrow[Var]\varphi$.\;\footnote{The order of $[x]$ and $[y]$ in the equivalence does not matter, i.e., $[y][x]\varphi\leftrightarrow[Var]\varphi$ is not valid as well.}
  \end{center}
\end{fact}

\begin{proof}
It suffices to find an instance.  Consider the following model: 
\begin{center}
\begin{tikzpicture}
\node(a)[circle,draw,inner sep=0pt,minimum size=5mm] at (0,0) {$s_1$};
\node(b) [circle,draw,inner sep=0pt,minimum size=5mm] at (1,0){$s_2$};
\node(c)[circle,draw,inner sep=0pt,minimum size=5mm] at (2,0) {$s_3$};
\node(d)[circle,draw,inner sep=0pt,minimum size=5mm] at (3,0) {$s_4$};
\node(e)[circle,draw,inner sep=0pt,minimum size=5mm] at (4,0) {$s_5$};
\draw[->] (a) to  (b);
\draw[->](b) to  (c);
\draw[->](c) to  (d);
\draw[->](d) to  (e);
\draw[->](e) to  [in=330, out=30,looseness=5] (e);
\end{tikzpicture}
\end{center}
Assume that both $x,y$ have the sight $1$, and the actual situation is $(s_1,s_3)$. Also, the original $\Sigma$ and the indistinguishability relations are given by the $1$-sight ability and the actual situation of their positions. Now, $[x][y]K_xy$ holds, but $[Var]K_xy$ fails.
\end{proof}

As shown by the example in the proof above, in $[x][y]\varphi$, a first movement $[x]$ may make $x$  know more about the position of $y$, but $[Var]$ leaves no room for $x$ to know that.  In what follows, we continue to study   $\ELCR$ and leave the exploration on the simultaneous variant to another occasion. 

\section{Applications of \texorpdfstring{$\ELCR$}{} to Cops and Robbers}\label{sec:application}

Now we will use the formal framework to track players' knowledge, to describe the winning conditions for players, among others. In this part, we assume that the round-restriction $n\in\mathbb{N}$ imposed on the winning condition is big enough. 

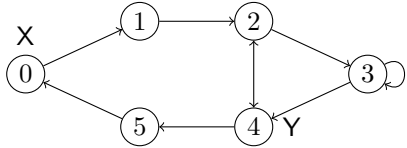
\vspace{1mm}

\begin{center}
\begin{tikzpicture}
\node(0)[circle,draw,inner sep=0pt,minimum size=5mm] at (0,0)[label=above:$\X$] {$0$};
\node(1)[circle,draw,inner sep=0pt,minimum size=5mm] at (1.5,.7) {$1$};
\node(2)[circle,draw,inner sep=0pt,minimum size=5mm] at (3,.7) {$2$};
\node(3) [circle,draw,inner sep=0pt,minimum size=5mm] at (4.5,0){$3$};
\node(4) [circle,draw,inner sep=0pt,minimum size=5mm] at (3,-.7)[label=right:$\Y$]{$4$};
\node(5) [circle,draw,inner sep=0pt,minimum size=5mm] at (1.5,-.7){$5$};
 
\draw[->] (0) to  (1);
\draw[->] (1) to  (2);
\draw[->](2) to  (3);
\draw[->](3) to  (4);
\draw[->](4) to  (5);
\draw[->](5) to  (0);
\draw[<->](2) to  (4);
\draw[->](3) to  [in=330, out=30,looseness=4] (3);
\end{tikzpicture}
\end{center}

 Let us revisit the game in Example \ref{ex1}. We will write $x$ and $y$ for Cop and Robber respectively. Recall that they have the sight $1$ and that in each round, $x$ acts first. As before, we write tuples for situations and usually highlight actual situations with underlines. For the example,  we denote by $\sigma_0=(0,4)$ the initial situation.

 \vspace{1mm}

 Based on the 1-sight ability,  $x$ knows her own position. Also, $x$ knows that $y$ is not at $0$, $1$ or $5$, but at $2$, $3$ or $4$. So, for player $x$, all of the following situations are possible: 
$$\boxed{x:\quad \underline{(0, 4)}, \; (0, 2), \; (0, 3)}$$
For $y$, he knows where himself is, but he cannot distinguish between the following:
$$\boxed{y:\quad \underline{(0, 4)}, \; (1, 4)}$$

All these aforementioned situations together form the set $\Sigma_0$ of possible situations of any models matching the game. 
 The indistinguishability relations  $\sim_{x}$ and $\sim_{y}$ among these situations are as described above. 

\vspace{1mm}

\begin{remark}\label{remark:high-order-knowledge}
    Here one can see that it is crucial to restrict ourselves to the setting without higher-order knowledge: if   higher-order knowledge is permitted, then  $\lr{K_x}K_yx$ holds at the actual situation, but w.r.t. the game, the formula is {\em not} a  correct description the knowledge of $x$.  Generalized proposals for settings with higher-order knowledge exist, but we leave these for future work. 
\end{remark}

\vspace{1mm}

Now, the game begins and player $x$  has only one option:

\vspace{1mm}

\noindent{\it $x$ moves to $1$.}\, The actual situation becomes $(1,4)$, written $\sigma_1$. 
Since they are not in the sight of each other, by the definition of the update, we can obtain the following:

 \vspace{1mm}

\begin{center}
 $\Sigma_1=\{\sigma\in \mathsf{R}^{x}(\Sigma_0|\sigma_0)\mid \sigma(x)\not\in \mathbb{D}^{1}(\sigma(y))\}=\{(1,4),\; (1,3)\}$.    
\end{center}

 \vspace{1mm}

\noindent Here $\Sigma_0|\sigma_0$ is exactly $\Sigma_0$. Again,  by the definition of updates, the indistinguishability relations among the situations  are as follows:

 \vspace{1mm}

\begin{center}
  $\boxed{x:\quad \underline{(1, 4)}, \; (1, 3)}\qquad\qquad  \boxed{y:\quad \underline{(1, 4)}}$  
\end{center}

 \vspace{1mm}

Recall the situations that cannot be distinguished by $x$ in the round 0. She now does not consider the case that $y$ is at $2$ to be possible: in the actual situation $(1,4)$ she cannot see directly $y$, but $(1, 2)$ is such a case.  Besides, for $y$, since $0$ is not a successor of any previous possibilities $0,1$ of the position of $x$ considered by him, $y$ does not consider the case  $(0,4)$  to be possible now, which is formally captured by the requirement that $\Sigma_1$ should be a subset of $\mathsf{R}^{x}(\Sigma_0|\sigma_0)$.

 \vspace{1mm}



Let us now move to the next stage:

\vspace{1mm}

\noindent{\it $y$ moves to vertex $5$.} \; We write $\sigma_2$ for the new situation $(1,5)$. Now  the players are not in the sight of each other as well, and the new class $\Sigma_2$ of situations and the indistinguishability relations given by our update policy are as  follows:

 \vspace{1mm}
 
    \begin{center}
 $\Sigma_2=\{\sigma\in \mathsf{R}^{y}(\Sigma_1|\sigma_1)\mid \sigma(x)\not\in \mathbb{D}^{1}(\sigma(y))\}=\{\underline{(1,5)},\;(1,3),\;(1,4)\}$. 

 \vspace{1mm}
 
  $\boxed{x:\quad \underline{(1,5)},\;(1,3),\;(1,4)}\qquad \qquad \boxed{y:\quad \underline{(1, 5)}}$  
\end{center}

\vspace{1mm}

Player $y$ still knows the actual situation, and player $x$ considers the unobservable successors of previous possibilities $3,4$ to be possible at this new stage.

 \vspace{1mm}

Finally, the game ends by the following step:

 \vspace{1mm}

\noindent{\it $x$ moves to $2$.}\; Now the actual situation, i.e., $(2,5)$, is the only possibility considered by each of the players, which can be seen from the following:

\vspace{1mm}

    \begin{center}
 $\{\sigma\in \mathsf{R}^{x}(\Sigma_2|\sigma_2)\mid \sigma(x)\not\in \mathbb{D}^{1}(\sigma(y))\}=\{\underline{(2,5)}\}$. 
 \end{center}

 \vspace{1mm}

Although the final $\{(2,5)\}$ is a singleton set, it is  obtained not by the clause $(a)$ in Definition \ref{def:updates}, but  by the clause $(b)$: this indicates that the realization of the actual situation depends on their indirect reasoning (with the knowledge about the graph structure and the $k$-sight ability), but not on the direct observation.  

 \vspace{1mm}
 
Generally speaking, the language can also be used to verify whether a player has a winning strategy. For instance, when the round-restriction $n$ is $2$, Cop has a winning strategy iff the following is true at the initial situation:

 \vspace{1mm}

\begin{center}
 $K_xy\lor\lr{x}K_xy\lor \lr{x}[y]K_xy \lor \lr{x}[y]\lr{x}K_xy\lor \lr{x}[y]\lr{x}[y]K_xy$. \footnote{For the setting without any round restrictions, the syntactic length would become infinite, which makes the formula not well-defined. To analyze this, one may need some enhanced tools such as  {\em the modal logic for substitution} \cite{substitution} that is motivated by the perfect information version of the game \cite{LHS-journal}.}   
\end{center}

 \vspace{1mm}

\noindent This suggests another potential application of the logic to the game: the model-checking problem for the formula of $\ELCR$ corresponds to the verification of existence of winning strategies of the players, and a complexity study would provide us an upper bound of the complexity of deciding this imperfect information game. As proved in \cite{elcr}, the static fragment $\ELCR^-$ has a $\mathsf{P}$-complete model-checking problem, and it remains to determine the exact complexity of  the model-checking problem for $\ELCR$.

\section{Some properties of \texorpdfstring{$\ELCR$}{}}\label{sec:logical-properties}
 
Having seen the basics  of the logic, we now turn to exploring some of its properties, involving validities and the effects of dynamic operators. Let us start by pointing out that many valid schemata of $\mathsf{FOL}$ fail in this new setting. For instance, formula 
$t_1\equiv t_2\to (\varphi\leftrightarrow \varphi[t_1/t_2])$
 is {\em not} a validity of $\ELCR$, where $\varphi[t_1/t_2]$ is obtained by replacing $t_2$ in $\varphi$ with $t_1$.\footnote{It is not hard to find a counterexample to, e.g., $c\equiv y\to (K_xc\leftrightarrow K_xy)$: $y$ is at $c$,  $x$ knows the value of $c$, but $x$   may not know where $y$ is. }  But when we restrict the formula $\varphi$ to those not involving knowledge, the  schema still holds.

 For any natural number $n\in\mathbb{N}$ and terms $t_1,t_2$, we define the following, which expresses the distance between $t_1$ and $t_2$:

 \vspace{1mm}

\begin{center}
  $\mathsf{D}^0t_1t_2:=t_1\equiv t_2$,

  \vspace{1mm}
  
  $\mathsf{D}^{n+1}t_1t_2:= 
\mathsf{D}^{n}t_1t_2 \lor \bigvee_{t\in Var\cup Cons}(\mathsf{D}^{n}t_1t\land (Rtt_2\lor Rt_2t))$.
\end{center}

 \vspace{1mm}

\noindent Since $Var\,\cup\, Cons$ is finite, the formulas above are well-defined. Now  $\mathsf{D}^{k}zt$ 
 means that the position of $t$ is in the sight of $z$.

  \vspace{1mm}

 \begin{fact}\label{fact:validities}
The following formulas  are  valid w.r.t.   $k$-sight models: 

\vspace{1mm}

\begin{itemize}

\item[$(1)$] $\mathsf{D}^kzt\to K_zt$, given that $z\in Var$. 

\vspace{1mm}

\item[$(2)$]  $(K_z\T\land P(t_1,t_2,\ldots,t_m))\to K_z P(t_1,t_2,\ldots,t_m)$, given that $\{t_1,\ldots,t_m\}\subseteq \T$. 
 \end{itemize}
 \end{fact}

  \vspace{1mm}

 
Formula $(1)$  means if $t$ is in the sight of $z$, then $z$ knows where $t$ is, and $(2)$ indicates that knowing the values of $\T$ means knowing the atomic facts involving $\T$: when $P$ is $R$, it characterizes the assumption that {\em players know the structures of game graphs}.  

In terms of the validities of knowledge operators $K_z\alpha\in\L$, although  $K_z\alpha\to\alpha$ is always the case, we do not have the ordinary $K_z\alpha\to K_zK_z\varphi$ or $\neg K_z\alpha\to K_z\neg K_z\varphi$  for {\em the positive reflection} and {\em the negative reflection}, respectively, since syntactically we restrict ourselves to the setting without high-order  knowledge. However, we can {\em mimic} them at the semantic level, in the sense of the following:

 \vspace{1mm}

\begin{fact}\label{fact:trans-sym}
 Let $M=(\bD,\bI,\Sigma,\sim)$, $\sigma\in \Sigma$ and $K_z\alpha\in\L$. 

 \vspace{1mm}
 
 \begin{itemize}
     \item[$(1)$] $M,\sigma\models K_z\alpha$ iff  for any $\sigma'\in\Sigma$ such that $\sigma\sim_z\sigma'$,  $M,\sigma'\models K_z\alpha$.

      \vspace{1mm}
     
     \item[$(2)$]  $M,\sigma\models \neg K_z\alpha$ iff for any $\sigma'\in\Sigma$ such that $\sigma\sim_z\sigma'$,  $M,\sigma'\models \neg K_z\alpha$.

     \vspace{1mm}

     \item[$(3)$] As a consequence, for any $\varphi$ of the form $K_z\alpha_1\land\dots K_z\alpha_n\land \lr{K_z}\beta_1\land\dots\land\lr{K_z}\beta_m$ with $1\le m+n$,  $M,\sigma\models \varphi$ iff for any $\sigma'\in\Sigma$ s.t. $\sigma\sim_z\sigma'$,  $M,\sigma'\models \varphi$.
 \end{itemize}
\end{fact}
\begin{proof}
    The key reason for these is that $\sim_z$ is an equivalence. We skip the details. 
\end{proof}

Also, by induction on $\LB$-formulas, we can show the following:

  \vspace{1mm}

\begin{fact}\label{fact:Boolean-locality}
    Let $M=(\bD,\bI,\Sigma,\sim)$ and $M'=(\bD,\bI,\Sigma',\sim')$ be models, and $\alpha\in\LB$. Let $\sigma\in\Sigma, \sigma'\in\Sigma'$ s.t. for any variable $z$ occurring in $\alpha$, $\sigma(z)=\sigma'(z)$. Then, it holds that: 
    
     \vspace{1mm}
    
    \begin{center}
        $M,\sigma\models\alpha$ iff $M',\sigma'\models\alpha$.
    \end{center}
\end{fact}
 
\vspace{1mm}

In what follows, for any $\T\subseteq Cons$ and $u\in\term$, we define the following:

 \vspace{1mm}

\begin{center}
    $Ru=\T:=\bigwedge_{t\in \T}Rut\land \bigwedge_{t\in Cons\setminus \T}\neg Rut$ 
\end{center}

 \vspace{1mm}

\noindent expressing that $\T$ is exactly the set of (the names of) $\bR$-successors of $z$. For instance,  $Rc_1=\{c_2\}$ means that via the relation $\bR$, $c_1$ can only reach $c_2$ but not all other constants. Now we have the following about the effects of updates on  $\LB$-formulas:

 \vspace{1mm}

\begin{fact}\label{fact:Boolean-connection}
   Let $M=(\bD, \bI, \Sigma, \sim)$ be a model, $\sigma_1\in\Sigma$. For any $\alpha\in\LB$ and  $z\in\{x,y\}$,  

    \vspace{1mm}
   
   \begin{center}
       $M,\sigma_1\models \bigwedge_{\T\subseteq Cons}(Rz=\T\to \bigwedge_{c\in \T}\alpha[c/z])$\; iff\; for all $\sigma_2\in\mathsf{R}^z(\sigma_1)$, $M',\sigma_2\models \alpha$.
   \end{center}

    \vspace{1mm}
   
\noindent   where $M'=(\bD, \bI, \Sigma', \sim')$  is the update of $M$ (associated to $\sigma_2$) induced by $[z]$.   
\end{fact} 
\begin{proof}
Notice that the $\T$ making $Rz=\T$ true is exactly $\{c\in Cons\mid M,\sigma_1\models Rzc\}$. Given this $\T$, we can fix $\mathsf{R}^z(\sigma_1)$ as $\{\sigma_1[z:=\bI(c)]\mid c\in \T\}$, where $\sigma_1[z:=c]$ might be different from $\sigma_1$ on the value of $z$: the former assigns the value of $c$ to $z$. Now,  the fact can be proved  by induction on formulas. The cases for Boolean connectives are trivial. The proofs for  $\pt$ and $t_1\equiv t_2$ are similar, and we just consider the former. 
\begin{center}
\begin{tabular}{r@{\qquad}c@{\qquad}l}
  $M,\sigma_1\models \bigwedge_{c\in \T} \pt[c/z]$   &  iff & for all $c\in \T$,   $M,\sigma_1[z:=\bI(c)]\models  \pt$      \\
  &  iff  &   for all $\sigma_2\in \mathsf{R}^z(\sigma_1)$,   $M',\sigma_2\models  \pt$    
\end{tabular}
\end{center}
The second holds by  the connection between $\T$ and $\mathsf{R}^z(\sigma_1)$ and the fact that all of $\bD$, $\bI$ and the value of the other variable are the same in $M,\sigma_1$ and $M',\sigma_2$.
\end{proof}

While the above preliminaries concern the static part $\L^{-}$, the result below is involved with the dynamic operators, which means that after a movement of a player, the player  still remembers what was known in the previous situation (Remark \ref{remark:update}):

 \vspace{1mm}

\begin{fact}\label{fact:movement-theother}
    Let $M=(\bD,\bI,\Sigma,\sim)$ be a model, $\sigma\in\Sigma$, and $\{z,z'\}=\{x,y\}$. If $M,\sigma\models K_zz'$, then $M,\sigma\models [z]K_zz'$.  
\end{fact}
\begin{proof}
W.l.o.g., let $z:=x$ and $z':=y$. Suppose $M,\sigma\models K_xy$. Now, let $\sigma'\in \mathsf{R}^{x}(\sigma)$ and we consider the updated model $M'$ associated to $\sigma'$. There are different cases.

\vspace{1mm}

First, if $\sigma'(x)\in\mathbb{D}^{k}(\sigma'(y))$, then the $\Sigma'$ of $M'$ is  $\{\sigma'\}$, which gives us $M',\sigma'\models K_xy$.

\vspace{1mm}

Next, assume that $\sigma'(x)\not\in\mathbb{D}^{k}(\sigma'(y))$. Now, let $\sigma'_1\in \Sigma'$ s.t. $\sigma'\sim'_x\sigma'_1$. Given $\sigma'_1\in \Sigma'$,  there is some $\sigma_1\in \Sigma|\sigma$ s.t. $\sigma'_1\in\mathsf{R}^{x}(\sigma_1)$. Notice that $\sigma_1(y)=\sigma'_1(y)$ and $\sigma(y)=\sigma'(y)$. So, it suffices to show that $\sigma_1(y)=\sigma(y)$, which then  can give us $\sigma'_1(y)=\sigma'(y)$.
Since $\sigma_1\in \Sigma|\sigma$, we have $\sigma_1\sim_x\sigma$ or $\sigma_1\sim_y\sigma$: if $\sigma_1\sim_x\sigma$, then by $M,\sigma\models K_xy$, it holds that $\sigma_1(y)=\sigma(y)$, otherwise by the $0\le k$-sight ability, $\sigma_1(y)=\sigma(y)$ holds as well. 
\end{proof}

 \section{Axiomatization of the static  fragment: \texorpdfstring{$\ELCR^-$}{}}\label{sec:axiomatization-static}

 Having seen the basic properties of the logic, we will proceed to  provide a Hilbert-style calculus for $\ELCR$ and show that it has a decidable satisfiability problem. In this section we will first consider the static part without action modalities $[x]$ or $[y]$. In what follows, we use $z,z'$ for variables in $\{x,y\}$. Also, for any $\T\subseteq Cons$, we define

 \vspace{1mm}

\begin{center}
$K_zz'=\T:=\bigwedge_{c\in \T} \lr{K_z}z'\equiv c \land \bigwedge_{c\in Cons\setminus \T} K_z\neg z'\equiv c$    
\end{center}

 \vspace{1mm}

\noindent expressing that $\T$ consists of the possible positions of $z'$ considered by $z$. 

\begin{table}[ht]
    \centering
    \caption{Proof system ${\bf ELCR}^-$ for the static part $\ELCR^-$.}
 \label{table:proof-system-1}
    \begin{tabular}{|ll|}
    \hline
    \multicolumn{2}{|l|}{~{\bf I:}\;General axioms and rules for Boolean connectives and $\equiv$}\\
     \hline
         $(\mathtt{Tau})$ &    Propositional tautologies \\
$(\mathtt{A1})$ & $t_1\equiv t_1$  \\
$(\mathtt{A2})$ & $t_1\equiv t_2\to t_2\equiv t_1$\\
$(\mathtt{A3})$& $t_1\equiv t_2\land t_2\equiv t_3\to t_1\equiv t_3$\\
$(\mathtt{A4})$& $t_1\equiv t_2\to (\alpha\leftrightarrow \alpha[t_1/t_2])$ given that $\alpha\in\LB$.\\
$(\mathtt{MP})$& From $\varphi\to\psi$ and $\varphi$, infer $\psi$.  \\ 
     \hline
    \multicolumn{2}{|l|}{~{\bf II:}\; Axioms for basics of the games}\\
    \hline
    $(\mathtt{Seriality})$  & $\bigvee_{t\in Cons}Rct$\\
    $(\mathtt{At}$-$\mathtt{Some}$-$\mathtt{Where})$  &    $\bigvee_{c\in Cons} z\equiv c$, given $z\in\{x,y\}$ \\
    $(k$-$\mathtt{sight})$ & 
    $\mathsf{D}^kzt\to K_zt$, where $z\in\{x,y\}$ and $t\in\term$.\\
    \hline
 \multicolumn{2}{|l|}{~{\bf III:}\; Axioms and rules for $K_z\alpha$}\\
     \hline
 $(\mathtt{K})$ & $K_z(\alpha\to \beta)\to(K_z\alpha\to K_z\beta)$, where $\alpha,\beta\in\LB$.\\ 
    $(\mathtt{T})$  & $K_z\alpha\to \alpha$, where $\alpha\in\LB$.\\
   $(\mathtt{Knowledge}$-$\mathtt{Ground})$    
      & $K_z z'=\T\to (K_z\alpha\leftrightarrow \bigwedge_{c\in \T}\alpha[c/z'])$,\\
      &\quad  given $\T\subseteq Cons$, $\{z,z'\}=\{x,y\}$ and $\alpha\in\LB$.\\
$(\mathtt{K}$-$\mathtt{Additivity})$    &  From $\varphi\to\alpha$, infer $\varphi\to K_z\alpha$, where $\varphi$ is of the form \\
&\quad  $K_z\alpha_1\land\dots\land K_z\alpha_n\land \lr{K_z}\beta_1\land\dots\land\lr{K_z}\beta_m$ \\
&\quad s.t. $1\le m+n$ and $\alpha,\alpha_{1\le i\le n},\beta_{1\le i\le m}\in\LB$.\\
$(\mathtt{K}$-$\mathtt{Elimination})$   &  From $\varphi\to (K_z\alpha\to\beta)$, infer $\varphi \to (\alpha\to  \beta)$, where $\varphi$ has\\
& \quad   the form $K_{z'}\alpha_1\land\dots\land K_{z'}\alpha_n\land \lr{K_{z'}}\beta_1\land\dots\land\lr{K_{z'}}\beta_m$ \\
&\quad with  $1\le m+n$,  $\alpha,\beta,\alpha_{1\le i\le n},\beta_{1\le i\le m}\in\LB$ and $\{z,z'\}=\{x,y\}$.  \\
              \hline
 \multicolumn{2}{|l|}{~{\bf IV:}\; Interaction axioms for  $K_z\alpha$ and $K_zt$}\\
     \hline
  $(\mathtt{De}$-$\mathtt{Re}$-$\mathtt{Knowledge})$   &  $t\equiv c\to (K_zt\leftrightarrow K_zt\equiv c)$, where $z\in\{x,y\}$ and $t\in\term$.\\
  $(\mathtt{Structure}$-$\mathtt{Knowledge})$   & $(K_z\T\land \alpha(t_1,t_2,\ldots,t_m))\to K_z \alpha(t_1,t_2,\ldots,t_m)$,\\
  &\quad  given that $\{t_1,\ldots,t_m\}\subseteq \T$ and $\alpha\in\LB$.\\
     \hline
    \end{tabular}
\end{table}

The details of the proof system are given in Table \ref{table:proof-system-1}. Notice that some axioms are redundant, but we keep them for convenience. Let us briefly comment on the axioms.

      \vspace{1mm}
     
\begin{itemize}
    \item[$\bullet$] In the part ${\bf II}$,   $(\mathtt{Seriality})$ means a state in a graph always has  ${\bf R}$-successors, and $(\mathtt{At}$-$\mathtt{Some}$-$\mathtt{Where})$ means that a player is always at some vertex in game graphs.

         \vspace{1mm}

    \item[$\bullet$] The part ${\bf III}$ is about knowledge operators. 
Formula  $(\mathtt{Knowledge}$-$\mathtt{Ground})$ lays out the foundation of knowledge: when $\T$ is the `uncertainty scope' of $z$ regarding the position of the other player $z'$, $z$ knows $\alpha$ iff every possible position of $z'$ considered by $z$ together with other relevant parameters can satisfy $\alpha$. 

         \vspace{1mm}

        \item[$\bullet$] The last part ${\bf IV}$ is about the interactions between $K_zt$ and $K_z\alpha$. The axiom $(\mathtt{De}$-$\mathtt{Re}$-$\mathtt{Knowledge})$, motivated by \cite{knowing-value}, states that when $t$ and $c$ have the same value,  $z$ knows the value of $t$ iff $z$ knows the fact that they have the same value. Finally, the axiom   $(\mathtt{Structure}$-$\mathtt{Knowledge})$ means that when $z$ knows the values of terms in a true statement $\alpha$, $z$ knows the fact that $\alpha$. 
\end{itemize}

 \vspace{1mm}

We write $\vdash\varphi$ if {\em $\varphi$ is provable in the calculus ${\bf ELCR}^-$}. Also, we write $\Phi\vdash\varphi$ if there exist finitely many $\varphi_1,\dots,\varphi_n\in\Phi$ such that $\vdash\varphi_1\land\dots\varphi_n\to \varphi$, and  when $\Phi$ is a singleton, e.g., $\{\psi\}$, we employ $\psi\vdash\varphi$ for simplicity.

 \vspace{2mm}

\begin{fact}\label{fact:soundness-static}
${\bf ELCR}^{-}$ is sound for $\ELCR^-$.
\end{fact}

\begin{proof}
    See Appendix \ref{sec:appendix-proof-for-soundness-static}.
\end{proof}

\begin{fact}\label{fact:provable-theorem}
The following formulas are provable and rules are derivable: 

 \vspace{1mm}

  \begin{itemize}

      \item[$(1)$] $\vdash K_zz$ and $\vdash K_zz\equiv z$. 

    \vspace{1mm}

      \item[$(2)$] From $\alpha$, infer $K_z\alpha$, where $\alpha\in\LB$.\quad $(\mathtt{Nec})$ 

    \vspace{1mm}
  
    \item[$(3)$] $\vdash K_zc$ for any $c\in Cons$.

    \vspace{1mm}

        \item[$(4)$] $\vdash K_z Cons\cup\{z\}$

    \vspace{1mm}

    \item[$(5)$]   $\vdash \alpha\to K_z \alpha$, given that $\alpha\in\LB$ does not contain the other variable.
    \vspace{1mm}
    \item[$(6)$]   $\vdash K_zz' \land \alpha \to K_z \alpha$, given that $\{z,z'\}=\{x,y\}$ and $\alpha\in\LB$.
        \vspace{1mm}
\item[$(7)$] From $K_z\alpha\to\beta$, infer $\alpha\to \beta$, where $\alpha,\beta\in\LB$. 
\end{itemize}  
\end{fact}
\begin{proof}
We prove the first three items, and  the others are  omitted to save space.

    \vspace{1mm}

(1) For the first item, it goes as follows:
    \begin{center}
\begin{tabular}{l@{\qquad}l@{\qquad}l}
  (1)   & $z\equiv z\to (K_zz\leftrightarrow K_zz\equiv z)$  &    $(\mathtt{De}$-$\mathtt{Re}$-$\mathtt{Knowledge})$    \\
  (2)   & $z\equiv z$  & $(\mathtt{A1})$\\
  (3)   &  $z\equiv z\to \mathsf{D}^{k}zz$ & (by  the definition of $\mathsf{D}^{k}t_1t_2$)\\
  (4)   & $\mathsf{D}^{k}zz$  & ($\mathtt{MP}$, 2, 3)\\
    (5)   & $\mathsf{D}^{k}zz\to K_zz$  & ($k$-$\mathtt{sight}$)\\
  (6)   & $K_zz$ & ($\mathtt{MP}$, 4, 5)\\
    (7)   & $K_zz\equiv z$ & (propositional logic, 1, 2, 6)
\end{tabular}
    \end{center}

     \vspace{1mm}

    (2). Assume that $\alpha$ is the case. Then, $K_zz\equiv z\to \alpha$. So, by  $(\mathtt{K}$-$\mathtt{Additivity})$, we have $K_zz\equiv z\to K_z\alpha$.  By the first item,  $K_zz\equiv z$. Using $(\mathtt{MP})$, we can obtain $K_z\alpha$.

 \vspace{1mm}

(3) For the third one, the details are as follows:
    \begin{center}
\begin{tabular}{l@{\qquad}l@{\qquad}l}
  (1)   & $c\equiv c\to (K_zc\leftrightarrow K_zc\equiv c)$  &    $(\mathtt{De}$-$\mathtt{Re}$-$\mathtt{Knowledge})$    \\
  (2)   & $c\equiv c$  & $(\mathtt{A1})$\\
  (3)   &  $K_zc\leftrightarrow K_zc\equiv c$ & ($\mathtt{MP}$, 1, 2)\\
  (4)   & $K_z c\equiv c$  & ($\mathtt{Nec}$, 2)\\
  (5)   & $K_zc$ & ($\mathtt{MP}$, 3, 4)
\end{tabular}
    \end{center}








This completes the proof.
\end{proof}

We say that a set $\Delta$ of formulas is {\em inconsistent} if 
there exists a finite set $\Gamma \subseteq \Delta$ such that $\vdash\bigwedge \Gamma  \to \bot$, and that $\Delta$ is {\em consistent} if it is not inconsistent.  
We also say that $\Delta$ is a {\em maximally consistent} set ($MCS$ for short) if $\Delta$ is  consistent and it holds that $\varphi \in \Delta$ or $\lnot \varphi \in \Delta$ for every formula $\varphi$ without the dynamic operators. 

Let $\Delta_1,\Delta_2\in MCS$. We write $\Delta_1=^{Cons}\Delta_2$ if they contain the same $\LB$-formulas $\alpha$ whose terms are constants. Also, let  $\T\subseteq \term$, and we write $\Delta_1=_{\T}\Delta_2$ if

 \vspace{1mm}

\begin{center}
    for any $t\in \T$ and $c\in Cons$, $t\equiv c\in \Delta_1$ iff $t\equiv c\in \Delta_2$. 
\end{center}

 \vspace{1mm}

 \noindent It is easy to see that when $\T\subseteq Cons$, $\Delta_1=^{Cons}\Delta_2$ implies $\Delta_1=_{\T}\Delta_2$.

  \vspace{1mm}

\begin{fact}\label{fact:term-equivalence}
Let   $\T\subseteq \term$ and  $\Delta_1,\Delta_2\in MCS$ s.t.  $\Delta_1=^{Cons}\Delta_2$ and $\Delta_1=_{\T}\Delta_2$. Also, let   $\alpha\in\LB$ s.t. the terms occurring in it are among $\T$. Then,  
    $\alpha\in\Delta_1$ iff $\alpha\in\Delta_2$. 
\end{fact}

\vspace{1mm}

It can be proved by  induction on $\alpha\in\LB$. Now we move to introducing the definition of the canonical models: 

 \vspace{1mm}

\begin{definition}\label{def:canonical-models}
Let $\Delta_0$ be a $MCS$ containing the following:

\vspace{1mm}

\begin{center}
$K_xy=\T_y,\; K_yx=\T_x,\; x\equiv c_x,\; y\equiv c_y$,    
\end{center}

\vspace{1mm}

\noindent where  $\T_y\cup \T_x\cup\{c_x,c_y\}\subseteq Cons$. We defined {\em the canonical model $M^{\Delta_0}=(\bD,\bI,\Sigma,\sim)$ induced by $\Delta_0$} in the following manner:

      \vspace{1mm}
     
\begin{itemize}
    \item[$\bullet$] $\bD=\{[c]\mid c\in Cons\}$, where $[c]:=\{c'\in Cons\mid c'\equiv c\in\Delta_0\}$

         \vspace{1mm}
         
    \item[$\bullet$] $\bI$ is defined as follows:  

         \vspace{1mm}
    
    $\bI(P):=\{([c_1],\dots,[c_n])\mid P(c_1,\dots,c_n)\in \Delta_0 \}$\quad and \quad
        $\bI(c):=[c]$.

       \vspace{1mm}

 \item[$\bullet$] $\Sigma$ is given by the following:\vspace{-2mm}
 \begin{center}
 {\small
 \begin{tabular}{c}
 $\{\Delta_0\}\cup$\\
 $\{\Delta_0=^{Cons}\Delta\in MCS \mid  K_xy=\T_y \in\Delta,   K_yx\equiv c_x\in\Delta,  y\equiv c, \neg c\equiv c_y\in\Delta \; \& \; c\in\T_y\}\cup$\\
      $\{\Delta_0=^{Cons}\Delta\in MCS \mid  K_yx=\T_x\in \Delta, K_xy\equiv c_y\in\Delta,  x\equiv c, \neg c\equiv c_x\in\Delta \; \& \; c\in\T_x  \}$
 \end{tabular}
 }
  \end{center}

\vspace{1mm}
  
\noindent such that for any $z\in\{x,y\}$, $\Delta(z)=[c]$ if $z\equiv c\in\Delta$.

     \vspace{1mm}

    \item[$\bullet$] For any $\Delta_1,\Delta_2\in\Sigma$,  
   $\Delta_1\sim_z\Delta_2$ iff $K_z\Delta_1=K_z\Delta_2$, where for any $\Delta\in \Sigma$, $K_z\Delta:=\{\alpha\mid K_z\alpha\in\Delta\}$.
\end{itemize}
\end{definition}

 \vspace{1mm}

The set $\Sigma$ defined above is desired, in the sense of the following:

 \vspace{1mm}

\begin{fact}\label{fact:canonical-assignment-set}
    Let $\Delta_0\in MCS$ s.t.
$K_xy=\T_y,\; K_yx=\T_x,\; x\equiv c_x,\; y\equiv c_y\in\Delta_0$,
 where  $\T_y\cup \T_x\cup\{c_x,c_y\}\subseteq Cons$. Let $M^{\Delta_0}=(\bD,\bI,\Sigma,\sim)$ be the induced canonical model.

 \vspace{1mm}
 
 \begin{itemize}
     \item[$(1)$] Let $\Gamma:=\{\alpha\in\LB\mid \alpha\in\Delta_0 \, \textit{ s.t. terms in }\,\alpha\;\textit{are constants}\}$. For any $c\in \T_y$ with $\neg c\equiv c_y\in\Delta_0$,   
       $\{K_xy=\T_y, \; y\equiv c, \; \neg c\equiv c_y \;\& \; K_yx\equiv c_x\}$  
 is consistent with $\Gamma$. For any  $c\in \T_x$ with $\neg c\equiv c_x\in\Delta_0$, 
         $\{K_yx=\T_x, \; x\equiv c, \;\neg c\equiv c_x\;\& \; K_xy\equiv c_y\}$  is consistent with $\Gamma$. So, when those constants $c$ exist, there are $MCS$s of $\Sigma$ containing the sets.

\vspace{1mm}

     \item[$(2)$] For any $\Delta\in\Sigma$, each of $x,y$ has exactly one value.

     \vspace{1mm}

         \item[$(3)$]   For any $\Delta_1,\Delta_2\in\Sigma$, if $\Delta_1(x)=\Delta_2(x)$ and $\Delta_1(y)=\Delta_2(y)$, then $\Delta_1=\Delta_2$.  
 \end{itemize}
\end{fact}

\begin{proof}
    The second item is easy to see. We merely prove the first and the third.

     \vspace{1mm}

    (1) For the first item, it suffices to consider the first part, and the second part is similar. Let $c\in \T_y$. Suppose for reductio that the set is not consistent with $\Gamma$. Then,
    \begin{center}
        $K_xy=\T_y\land K_yx\equiv c_x\vdash \neg c\equiv c_y\to (\bigwedge\Phi\to\neg y\equiv c)$, \quad for some finite $\Phi\subseteq\Gamma$
    \end{center}
    Applying the rule $(\mathtt{K}$-$\mathtt{Elimination})$ can give us the following:
      \begin{center}
     $\vdash K_xy=\T_y\to  (x\equiv c_x\to (\neg c\equiv c_y\to (\bigwedge\Phi\to\neg y\equiv c)))$.
    \end{center}  
By $(\mathtt{K}$-$\mathtt{Additivity})$,  
     $\vdash K_xy=\T_y\to  K_x(x\equiv c_x\to (\neg c\equiv c_y\to (\bigwedge\Phi\to\neg y\equiv c)))$. Then, $\vdash K_xy=\T_y\to  (K_x x\equiv c_x\to (K_x\neg c\equiv c_y\to (K_x\bigwedge\Phi\to K_x\neg y\equiv c)))$. However, since $K_xy=\T_y\land  K_x x\equiv c_x \land K_x\neg c\equiv c_y\land K_x\bigwedge\Phi \in\Delta_0$, we have $K_x\neg y\equiv c$, which contradicts $c\in\T_y$.

     \vspace{1mm}

     (2) We now move to the third item. By the construction, it is easy to see that one of $\Delta_1$ and $\Delta_2$ is $\Delta_0$ iff $\Delta_1=\Delta_2=\Delta_0$. In what follows, we consider  the case that $\Delta_1 \not= \Delta_0$ and $\Delta_2 \not= \Delta_0$. Then, by construction, exactly one of $\Delta_1(x)\not=[c_x]$ and $\Delta_1(y)\not=[c_y]$ is the case, but they cannot hold at the same time. It suffices to consider  the case that $\Delta_1(x)=[c_x]$ and $\Delta_1(y)\not=[c_y]$, and the other case is analogous.

     By assumption, $\Delta_1=_{Var}\Delta_2$, and by construction, $K_xy=\T_y\in\Delta_1\cap \Delta_2$.  By the former and  Fact \ref{fact:term-equivalence}, the $\LB$-parts of $\Delta_1$ and $\Delta_2$  are the same.  Also, the arguments for the formulas   $K_xt$ and $K_yt$ are simple. We now move to considering  $K_x\alpha$ and $K_y\alpha$.
     
     Now suppose that $K_x\alpha\in\Delta_1$. Then, $\bigwedge_{t\in \T_y}\alpha[c_x/x][t/y]\in \Delta_1$. Again, by Fact \ref{fact:term-equivalence}, $\bigwedge_{t\in \T_y}\alpha[c_x/x][t/y]\in \Delta_2$, which together with $K_xy=\T_y, x\equiv c_x\in\Delta_2$ can give us $K_x\alpha\in\Delta_2$. The converse is similar.

     Next, suppose that $K_y\alpha\in\Delta_1$. Let $c'_y\in\T_y$ with $c'_y\equiv y\in\Delta_1$. So, by assumption, $c'_y\equiv y\in\Delta_2$. It is easy to see that $\alpha[c_x/x][c'_y/y]\in\Delta_1$. By Fact \ref{fact:term-equivalence}, $\alpha[c_x/x][c'_y/y]\in\Delta_2$.  With this, we know that $K_y\alpha\in\Delta_2$. Again, the converse is similar.
\end{proof}

\begin{theorem}
 For any $\Delta_0\in MCS$, the induced $M^{\Delta_0}$ is a $k$-sight model.
\end{theorem}

\begin{proof}
(1) We first show that the components are well-defined. It is simple to check that the definition itself does not depend on the choice of representatives of equivalence classes $[c]$. Also, as indicated by Fact \ref{fact:canonical-assignment-set}, $\Sigma$ is a well-defined set of assignments.

\vspace{1mm}

(2) By the axiom $(\mathtt{Seriality})$, the relation $\bI(R)$ is serial.

\vspace{1mm}

(3) It is straightforward to see that $\sim_x$ and $\sim_y$ are equivalence relations. 

\vspace{1mm}

(4) It remains to show that players have the $k$-sight ability. Assume that $\Delta_1\sim_x\Delta_2$, and we show that for any $z\in\{x,y\}$, if $\Delta_1(z)\in \mathbb{D}^{k}(\Delta_1(x))$, then $\Delta_1(z)=\Delta_2(z)$.

    \vspace{1mm}

By the definition of $\mathsf{D}^k$, we have $\mathsf{D}^kxz$. Then, by the axiom $(k$-$\mathtt{sight})$, it holds that $K_xz$. Also, by the axiom $(\mathtt{At}$-$\mathtt{Some}$-$\mathtt{Where})$, using the fact that $\Delta_1$ is a $MCS$, there is some $c\in Cons$ s.t. $z\equiv c\in\Delta_1$.  Now, using the axiom $(\mathtt{De}$-$\mathtt{Re}$-$\mathtt{Knowledge})$, we can obtain $K_xz\equiv c\in \Delta_1$. Using the axiom $(\mathtt{T})$, we have $z\equiv c\in\Delta_1$. Also, since $\Delta_1\sim_x\Delta_2$, it follows from $K_xz\equiv c\in \Delta_1$ that $z\equiv c\in\Delta_2$. So, $\Delta_1(z)=\Delta_2(z)=[c]$.   
\end{proof}

Now, we show the following {\em Existence Lemma}:

 \vspace{1mm}

\begin{lemma}\label{lemma:existence-lemma-Kphi}
 Let $\Delta_0$ be a $MCS$ such that
$K_xy=\T_y,\; K_yx=\T_x,\; x\equiv c_x,\; y\equiv c_y\in\Delta_0$,    
 where  $\T_y\cup \T_x\cup\{c_x,c_y\}\subseteq Cons$. Also, let $M^{\Delta_0}$ be the induced canonical model of $\Delta_0$. For any $z\in\{x,y\}$ and any $\Delta\in MCS$ from $M^{\Delta_0}$, when  $K_z\alpha\not\in\Delta$, there is some  $\Delta'$ from $M^{\Delta_0}$ s.t. $\Delta\sim_z\Delta'$ and $\alpha\not\in\Delta$. 
\end{lemma}

\begin{proof}
W.l.o.g., let $z:=x$.  When $\Delta(x)\not=[c_x]$, by construction, $K_xy\equiv c_y\in\Delta$. Now, $K_x\alpha\not\in \Delta$ gives us $\neg\alpha\in \Delta$, so a desired $\Delta'$ is $\Delta$ itself. Let us assume that $\Delta(x)=[c_x]$. Now, it follows from $K_z\alpha\not\in\Delta$ that there is $c^0_y\in\T_y$ such that $\neg \alpha[c_x/x][c^0_y/y]\in\Delta$.

If $c^0_y\equiv c_y\in\Delta$, then $\Delta=\Delta_0$ (cf. the proof for the item $(3)$ of Fact \ref{fact:canonical-assignment-set}).  Now, $\Delta_0$ itself is such a $\Delta'$. Let us now assume that $\neg c^0_y\equiv c_y\in\Delta$. We consider  $\Delta_1\in MCS$  such that $\Delta'=^{Cons}\Delta_0$ and $\{K_xy=\T_y, K_yx\equiv c_x,  y\equiv c^0_y, \neg c^0_y\equiv c_y\}\subseteq \Delta'$. We have see from the item $(1)$ of Fact \ref{fact:canonical-assignment-set} that $\Sigma$ does contain such a $\Delta_1$. It is simple to see that $\neg\alpha\in\Delta_1$. Now it remains to show that $\Delta\sim_x\Delta_1$, which is guaranteed by the fact $\Delta_1=^{Cons}\Delta_0$, $K_xy=\T_y\in\Delta\cap \Delta_1$ and the axiom $(\mathtt{Knowledge}$-$\mathtt{Ground})$.   
\end{proof}

As a consequence, we have the following:

 \vspace{1mm}

\begin{lemma}\label{lemma:existence-lemma-Kxt}
Let $\Delta_0\in MCS$ and $M^{\Delta_0}$ be its induced canonical model. For $\{z,z'\}=\{x,y\}$ and any $\Delta\in MCS$ from $M^{\Delta_0}$, when $\neg K_zz'\in\Delta$, there is a $\Delta'$ from $M^{\Delta_0}$ such that $\Delta'\sim_z\Delta$ and for some $c\in Cons$, $c\equiv z'\in\Delta$ and $c\not\equiv z'\in\Delta'$.
\end{lemma}
 
\begin{proof}
Since $K_zz'\leftrightarrow\bigwedge_{c\in Cons}(\lr{K_z}z'\equiv c\to  K_zz'\equiv c)$,  it holds by Lemma \ref{lemma:existence-lemma-Kphi}.
\end{proof}

Next we proceed to show the crucial {\em Truth Lemma}:

 \vspace{1mm}

 \begin{lemma}\label{lemma-truth-lemma}
    Let $\Delta_0\in MCS$ and $M^{\Delta_0}$ be its induced canonical model. For any $\varphi\in\L^{-}$ and $\Delta$ from $M^{\Delta_0}$,  
        $M^{\Delta_0},\Delta\models\varphi$ iff $\varphi\in \Delta$.
 \end{lemma}

\begin{proof}
It goes by induction on formulas. The cases for Boolean connectives $\neg,\land$ are straightforward by  induction hypothesis, and we consider  others. 

\vspace{1mm}

(1) Formula $\varphi$ is  $P(t_1,\dots,t_n)$. Then, we have the following  equivalences:
\begin{center}
    \begin{tabular}{r@{\quad}c@{\quad}l}
      $M^{\Delta_0},\Delta\models \varphi$   & iff   &  $(t_1^{(\bI,\Delta)},\dots, t_n^{(\bI,\Delta)})\in \bI(P)$\\
         &  iff &  $([c_1],\dots,[c_n])\in\bI(P)$\\
       &  iff & $P(c_1,\dots,c_n)\in\Delta_0$\\
         &  iff & $P(c_1,\dots,c_n)\in\Delta$  \\
         &iff &$P(t_1,\dots,t_n)\in\Delta$
    \end{tabular}
\end{center}
In the second equivalence, those $c_i$ are constants, and when $t_i$ is a constant, $c_i$ can be $t_i$ itself, and when $t_i$ is a variable, we put $c_i:=c$ for some $c\in Cons$ s.t. $c\equiv t_i\in\Delta$ (due to  $(\mathtt{At}$-$\mathtt{Some}$-$\mathtt{Where})$, such a constant $c$ always exists). By Fact \ref{fact:term-equivalence} and $\Delta=^{Cons} \Delta_0$ (the latter follows from $\Delta\in\Sigma$), the fourth equivalence holds. The last one holds directly (if all those $t_i$ are constants) or holds by $(\mathtt{A4})$ (if some of those $t_i$  are variables).

\vspace{1mm}

(2) Formula $\varphi$ is $t_1\equiv t_2$. The reasoning is similar to the case above.

\vspace{1mm}

(3)  Formula $\varphi$ is $K_zt$. When $t$ is a constant or $z$, we  have $K_zt\in\Delta$ (Fact \ref{fact:provable-theorem}), and meanwhile, for the semantic aspect, it follows from the construction of canonical models that $M^{\Delta_0},\Delta\models K_zt$. We now move to the case that $t$ is the other variable $z'$. 

By Lemma \ref{lemma:existence-lemma-Kxt}, when $M^{\Delta_0},\Delta\models K_zt$,  $K_zt\in\Delta$.  For the other direction,  assume that $K_zt\in\Delta$,   $\Delta\sim_z\Delta'$,  $\Delta(z')=[c_1]$ and $\Delta'(z')=[c_2]$. We will show $[c_1]=[c_2]$, for which it suffices to prove that $c_1\equiv c_2\in\Delta'$. Now,  $c_1\equiv z'\in\Delta$ and $c_2\equiv z'\in\Delta'$. Given  $c_1\equiv z'\in\Delta$ and $K_zt\in\Delta$, using $(\mathtt{De}$-$\mathtt{Re}$-$\mathtt{Knowledge})$ we obtain $K_zz'\equiv c_1\in\Delta$. Since $\Delta\sim_z\Delta'$, it holds that $z'\equiv c_1\in\Delta'$. So, $c_1\equiv c_2\in\Delta'$.

\vspace{1mm}

(4) Finally, we consider the case that $\varphi$ is $K_z\alpha$.  We proceed as follows:
\begin{center}
    \begin{tabular}{r@{\quad}c@{\quad}l}
$M^{\Delta_0},\Delta\models K_z\alpha$   & iff  & for all
$\Delta'\in\Sigma$, if $\Delta\sim_z\Delta'$, then $M^{\Delta_0},\Delta'\models \alpha$  \\
& iff  &  for all $\Delta'\in\Sigma$, if $\Delta\sim_z\Delta'$, then $\alpha\in\Delta'$   \\
& iff  & $K_z\alpha\in\Delta$   
    \end{tabular}
\end{center}
The second equivalence holds by induction hypothesis. One direction of the last equivalence holds by Lemma \ref{lemma:existence-lemma-Kphi}, and the converse   holds by $K_z\alpha\in\Delta$ and $\Delta\sim_z\Delta'$.  
\end{proof}

 Finally we can show the following strong completeness result for  ${\bf ELCR}^{-}$:

 \vspace{1mm}

\begin{theorem}\label{theorem:completeness-static}
   For any set $\Gamma\subseteq\L^{-}$, if $\Gamma$ is consistent, then it is satisfiable. As a consequence, the static part is compact. 
\end{theorem}

\begin{proof}
Given a consistent set $\Gamma\subseteq\L^{-}$, by a  Lindenbaum-style argument \cite{modallogic}, we can  extend it  to some $\Delta_0\in MCS$. Then, there is a canonical model $M^{\Delta_0}$ induced by $\Delta_0$. By Lemma \ref{lemma-truth-lemma}, for any $\Delta$ from $M^{\Delta_0}$  and any $\varphi\in\L^-$,  
        $M^{\Delta_0},\Delta\models\varphi$ iff $\varphi\in \Delta$. Notice that $\Delta_0$ is also an assignment in  $M^{\Delta_0}$. Since $\Gamma\subseteq \Delta_0$,  $M^{\Delta_0},\Delta_0\models\Gamma$, as desired. 
\end{proof}

Moreover, we can show the following:

 \vspace{1mm}

\begin{corollary}\label{coro:decidability-static}
 $\ELCR^{-}$ is decidable.
\end{corollary}

\begin{proof}
 Notice that in $M^{\Delta_0}=({\bf D}, {\bf I}, \Sigma,\sim)$, both  ${\bD}$ and $\Sigma$ are finite, so a static formula that is satisfiable can be satisfied by a finite model. Moreover, given that  the fragment is finitely axiomatizable, we can obtain its decidability immediately. 
\end{proof}

\section{Axiomatization  of \texorpdfstring{$\ELCR$}{}}\label{sec:axiomatization-whole-logic}

So far, we have shown a complete calculus for the static base $\ELCR^{-}$, and this section is devoted to providing a complete calculus for $\ELCR$. To do so, it is enough to find an effective way to eliminate the occurrences of dynamic operators, like the techniques of {\em recursion axioms} developed for  $\mathsf{DEL}$.

First, we consider the following formulas that can reduce $\mathcal{L}_{\mathsf{BD}}$ to $\LB$:
\begin{align*}
  [z]\alpha \leftrightarrow &   \bigwedge_{\T\subseteq Cons}(Rz=\T\to \bigwedge_{c\in \T}\alpha[c/z]), \quad \textit{given $\alpha\in \LB$} \tag{$\mathtt{R1}$}\\
         [z]\neg \varphi  \leftrightarrow &   \neg [z]\varphi \tag{$\mathtt{R2}$}\\
      [z](\varphi\land\psi)  \leftrightarrow &     [z]\varphi\land[z]\psi \tag{$\mathtt{R3}$}
\end{align*}
Notice that there is always a set $\T\subseteq Cons$ making $Rz=\T$ true, and when $\alpha$ does not contain any occurrence of $z$, $(\mathtt{R1})$ amounts to  $[z]\alpha\leftrightarrow \alpha$.\footnote{Recall that we have $(\mathtt{Seriality})$ as an axiom of the static part, so it cannot be the case that $[z]\bot$.} Moreover,  different from the ordinary format of recursion axioms for $\mathsf{DEL}$, the $\alpha$ in $(\mathtt{R1})$ need not be an atom.

Although the syntax of $\L_{\mathsf{BD}}$ allows nested occurrences of dynamic operators, we can always start with the elimination of some innermost occurrence of a dynamic operator, as the case for  $\mathsf{DEL}$.

\vspace{1mm}

Next, we move to dealing with $[z]K_{z'}t$:
\begin{align*}
 [z]K_{z'}t \leftrightarrow &  \top, \quad \textit{where} \; t\in Cons\cup\{z'\} \tag{$\mathtt{R4}$}\\ 
 [z]K_zz'  \leftrightarrow &   K_zz'\lor     \\
  &\bigwedge_{\T,\T_1\subseteq Cons}(K_z z'=\T \land Rz=\T_1 \to \tag{$\mathtt{R5}$}\\
     & \bigwedge_{a\in \T_1}(\mathsf{D}^kaz' \lor  \bigwedge_{b\in \T}(b\not\equiv z'\to \mathsf{D}^k ab))), \quad  \textit{where}\; \{z,z'\}=\{x,y\}\\
 [z]K_{z'}z  \leftrightarrow & \bigwedge_{\T,\T_1\subseteq Cons} (Rz=\T\land K_{z'} z= \T_1\to ((\bigwedge_{t\in \T}\mathsf{D}^kz't)\lor\\
 &  \bigwedge_{t_1,t_2\in \T_1,\;t'_1,t'_2\in Cons} (Rt_1t'_1\land Rt_2t'_2\land \tag{$\mathtt{R6}$}\\
 & \neg\mathsf{D}^kz't'_1\land \neg\mathsf{D}^kz't'_2\to t'_1\equiv t'_2))),\quad  \textit{given that}\; \{z,z'\}=\{x,y\} 
\end{align*}

\noindent The formula $(\mathtt{R4})$  suggests that for any $z,z'\in\{x,y\}$, $[z]K_{z'}Cons\cup\{z'\}$ is always the case. 
 By $(\mathtt{R5})$, for different $z,z'$,  $[z]K_zz'$ holds if, and only if, one of the following holds:

 \vspace{1mm}

\begin{itemize}
    \item[$(i)$] $z$ has already known the position of $z'$ before the movement (i.e., $K_zz'$).

\vspace{1mm}
    
    \item[$(ii)$]  Given the possible positions $\T$ of $z'$ considered by $z$ and the set $\T_1$ of $\bR$-successors of $z$, for any possible new position $a$ of $z$, {\em either}  $z$ can see directly where $z'$ is {\em or} all other vertices different from the position of $z'$ can be observed directly by $z$.
\end{itemize}

 \vspace{1mm}

\noindent The last $(\mathtt{R6})$ means that for different $z$ and $z'$, when we assume that $Rz=\T$ and $K_{z'}z=\T_1$, $[z]K_{z'}z$ is the case if, and only if, one of the following is the case:

 \vspace{1mm}

\begin{itemize}
    \item[$(i)$] All those of $\T$ are in the sight of $z'$.

    \vspace{1mm}
    
    \item[$(ii)$] When $t_1,t_2\in \T_1$  have successors that are not in the sight of $z'$, all those successors not in the sight of $z$ are the same.
\end{itemize}

 \vspace{1mm}

Now  we proceed to tackle $[z]K_{z'}\alpha$.
In what follows, for any $\alpha\in\LB$ and $z\in\{x,y\}$, we use $\alpha(z)$ to highlight that $z$ does occur in $\alpha$. The details are as follows:
\begin{align*}
 [z] K_{z}\alpha \leftrightarrow& [z]\alpha, \quad \textit{given that $\alpha\in \LB$ does not contain the other variable} \tag{$\mathtt{R7}$}\\
 [z] K_{z'}\alpha \leftrightarrow&  \alpha, \quad \textit{given $\{z,z'\}=\{x,y\}$ and $\alpha\in \LB$ does not contain $z$}  \tag{$\mathtt{R8}$}\\
[z]K_{z}\alpha(z')\leftrightarrow &  \bigwedge_{\T,\T_1\subseteq Cons}(Rz=\T \land K_zz'=\T_1\to ((K_zz'\land \bigwedge_{c\in \T}\alpha[c/z])\lor \\
  & (\neg K_zz'\land \bigwedge_{t\in \T}((\mathsf{D}^ktz'\land \alpha[t/z])\lor \tag{$\mathtt{R9}$}\\ 
  &(\neg \mathsf{D}^ktz'\land \bigwedge_{c\in \T_1}(\neg \mathsf{D}^ktc\to  \alpha[t/z][c/z']))))),\\
  & \textit{given $\{z,z'\}=\{x,y\}$ and $\alpha\in \LB$} \\
[z]K_{z'}\alpha(z)\leftrightarrow & \bigwedge_{\T,\T_1\subseteq Cons}(Rz=\T\land K_{z'}z=\T_1 \to  (\bigwedge_{t\in \T}\alpha[t/z]) \land  ((\bigwedge_{t\in \T}\mathsf{D}^kz't)\lor \\
 &\bigwedge_{t_1\in \T_1,\;t_2\in Cons}(Rt_1t_2\land \neg \mathsf{D}^kz't_2\to \alpha[t_2/z]))), \tag{$\mathtt{R10}$}\\
 &\textit{given $\{z,z'\}=\{x,y\}$ and $\alpha\in \LB$}  
\end{align*}

\noindent  Among these, the  formulas $(\mathtt{R7})$ and $(\mathtt{R8})$ are simple, and principles $(\mathtt{R9})$ and $(\mathtt{R10})$ are their  complements, respectively, for the more complicated settings: one can check that the latter two amount to the former ones when $\alpha$ does not contain the corresponding variables. For the formula $(\mathtt{R9})$, we assume that $Rz=\T$ and $K_zz'=\T_1$, and the formula expresses that $[z]K_z\alpha(z')$ iff one of the following is the case:

 \vspace{1mm}

\begin{itemize}
\item[$(i)$] $z$ already knew the position of $z'$ {\em before the movement}, and {\em after any movement of $z$}, $\alpha$ is the case.

\vspace{1mm}

\item[$(ii)$]  $z$ did not know the position of $z'$ {\em before the movement}, but for any $t$ where $z$ can move to, {\em either} $z'$ can be observed from $t$ and $\alpha$ is the case, {\em or} $z'$ cannot be observed from $t$ but any unobservable possibility in $\T$ from $t$ makes  $\alpha$ true.  
\end{itemize}

 \vspace{1mm}

\noindent Finally, for the principle $(\mathtt{R10})$, we assume that $Rz=\T$ and $K_{z'}z=\T_1$, and it states that $[z]K_{z'}\alpha(z)$ is true iff $\bigwedge_{t\in \T}\alpha[t/z]$ and one of the following holds:

 \vspace{1mm}

\begin{itemize}
    \item[$(i)$] $z$ can only move into the observable range of $z'$.

    \vspace{1mm}
    
    \item[$(ii)$] For any possibility $t_1\in \T_1$ (considered by $z'$ before the movement),  any of its unobservable $\bR$-successor from $z'$ makes $\alpha$ true.
\end{itemize}

 \vspace{1mm}

We write ${\bf ELCR}$ for {\em the resulting calculus obtained by adding $(\mathtt{R1})$-$(\mathtt{R10})$ to the proof system ${\bf ELCR}^{-}$}, and generalize the usages of $\vdash$ to the setting of ${\bf ELCR}$. We show that ${\bf ELCR}$ is both  sound and complete.

 \vspace{1mm}

\begin{theorem}\label{theorem:soundness}
    The calculus ${\bf ELCR}$ is sound for $\ELCR$.
\end{theorem}

\begin{proof}
    See Appendix \ref{sec:appendix-proof-for-soundness}.
\end{proof}

 \vspace{1mm}

\begin{theorem}\label{theorem:completeness}
   For any set $\Gamma\subseteq\L$, if $\Gamma$ is consistent, then it is satisfiable. As a consequence, $\ELCR$ is also compact. 
\end{theorem}

\begin{proof}
For any $\varphi\in \Gamma$, by  $(\mathtt{R1})$-$(\mathtt{R10})$, there is always a static $\psi\in\L^{-}$ such that $\vdash \varphi\leftrightarrow \psi$. Based on this, there is a set $\Gamma'\subseteq \L^{-}$ that is provably equivalent to $\Gamma$.  $\Gamma$ is consistent, so is $\Gamma'$. Then, by  Theorem \ref{theorem:completeness-static}, $\Gamma'$ is satisfiable. By the soundness (Theorem \ref{theorem:soundness}), $\Gamma$ is satisfiable, as desired. 
\end{proof}

Finally, the arguments in the proof above suggest the following:

 \vspace{1mm}

\begin{corollary}\label{coro:decidability-whole-logic}
$\ELCR$ is decidable.
\end{corollary}

\begin{proof}
By the proof for the completeness of {\bf ELCR}, any formula of $\L$ is equivalent to a static formula of $\L^{-}$. Now it follows from Corollary \ref{coro:decidability-static} that  $\ELCR$ is also decidable.
\end{proof}

So, different from the undecidable proposal in \cite{LHS-journal} for the perfect information game of Hide and Seek, we now have a decidable framework for the imperfect information version of the game.\footnote{As stated in \cite{LHS-journal}, the culprit of the undecidability of that logic  is the equality. Our language also contains the symbol, and one reason for the  decidability of $\ELCR$ is that we confine ourselves to the finite setting. But it is important to notice that generalizing $\ELCR$ into infinite case does not necessarily lead us to an undecidable framework. For more discussion on when the equality is dangerous, see \cite{hlhs-axiomatization}. } 

\section{Related approaches} \label{sec:related}

There are many other logical milestones  in the realm of knowledge dynamics, especially the paradigm of  $\mathsf{DEL}$. 
 A $\mathsf{DEL}$ approach to Cops and Robbers  is discussed in \cite{johan-del-game}, and we will formalize  its key points 
and  apply the ideas 
to the same Example \ref{ex1}. This would form a visual comparison between the two  frameworks,  indicating the succinctness of our proposal.   In addition, this section also offers an  overview on the recent logical   frameworks developed for games played on graphs and examines logics addressing various dependencies that are technically relevant to our work. 

\subsection{\texorpdfstring{$\mathsf{DEL}$}{}-approach to the game}\label{sec:del}

 While exploring $\ELCR$,  
one may wonder  whether 
some appropriate modifications to the techniques developed for $\mathsf{DEL}$,  e.g., {\em product update} \cite{BMS},  may apply to the game under consideration. As stated  above, some  such ideas are sketched in \cite{johan-del-game},  and  we now  formalize  them  with some necessary modifications. 
In this part, to avoid digressing too far, we will omit many basics of $\mathsf{DEL}$ and  discuss the key points  in a concise manner. 

  \vspace{1mm}

  Our main focus is on the product updates of {\em epistemic models} and {\em event models}. The former are exactly our $k$-sight models, and for the latter, the game graphs themselves serve as an important parameter, but with a different reading now: edges  are now moves from a vertex to another.
 \vspace{1mm}

 \begin{definition}\label{def:event-model}
     Let $M=(\bD,\bI,\Sigma,\sim)$ be a $k$-sight model and $\sigma\in\Sigma$. There are {\em event model $E^{x}_{\sigma}$ for $x$ associated to $\sigma$} and {\em event model $E^{y}_{\sigma}$ for $y$ associated to $\sigma$}, which are analogous. We just define the former, and skip that  of the latter. Details are as follows:

 \vspace{1mm}

     \begin{center}
         $E_{\sigma}^{x}=(\bI(R),\approx, \{pre_{(s,t)}\mid (s,t)\in\bI(R)\}, \{post_{(s,t)}\mid (s,t)\in\bI(R)\})$, \noindent where
     \end{center}

 \vspace{1mm}

\begin{itemize}
    \item[$\bullet$] For each $(s,t)\in\bI(R)$, its {\em pre-condition} $pre_{(s,t)}$ is that $x$ is at $s$, and its {\em post-condition} $post_{(s,t)}$ is that $x$ is at $t$, which together mean the change of the position of $x$ from $s$ to $t$ caused by the action. 

\vspace{1mm}
    
  \item[$\bullet$]  The {\em indistinguishability  relation $\approx_x$  for $x$} is the {\em identity relation} on the pairs $\bI(R)$, since $x$ always knows which movement she is taking.

\vspace{1mm}

 \item[$\bullet$] The case for the {\em indistinguishability  relation $\approx_y$  for $y$} is more complicated, since we need to take the $k$-sight ability into account, which determines whether or not $y$ knows the action of $x$:

\vspace{1mm}
 
    \begin{itemize}
 \item[$\bullet$] For those $(s,t)\in\bI(R)$  with $s,t\in \mathbb{D}^k(\sigma(y))$,   $\approx_y$ is the identity relation.

 \item[$\bullet$] For those $(s,t)\in\bI(R)$  with $s\in \mathbb{D}^k(\sigma(y))$ and  $t\not\in \mathbb{D}^k(\sigma(y))$,

\begin{center}
 $(s,t)\approx_y(s',t')$\quad iff \quad  $s=s'$  and $t'\not\in  \mathbb{D}^k(\sigma(y))$.
\end{center}

 \item[$\bullet$] For those $(s,t)$ with $s\not\in \mathbb{D}^k(\sigma(y))$ and  $t\in \mathbb{D}^k(\sigma(y))$,

        \begin{center}
 $(s,t)\approx_y(s',t')$\quad iff \quad  $t=t'$  and $s'\not\in  \mathbb{D}^k(\sigma(y))$.
\end{center}

 \item[$\bullet$] For those  $(s,t)$ with $s\not\in \mathbb{D}^k(\sigma(y))$ and  $t\not\in \mathbb{D}^k(\sigma(y))$,

\begin{center}
 $(s,t) \approx_y (s',t')$\quad iff \quad  $s,t,s',t'\not\in  \mathbb{D}^k\sigma(y))$.
\end{center}
    \end{itemize}
\end{itemize}  
 \end{definition}

  \vspace{1mm}

 \begin{definition}\label{def:product}
     Let  $M=(\bD,\bI,\Sigma,\sim)$ be a $k$-sight model, $\sigma\in\Sigma$ and $E_{\sigma}^{x}$ be the event model associated to $\sigma$. The {\em product update} $(M,\sigma)\times E^{x}_{\sigma}$ is a new $k$-sight model $(\bD,\bI,\Sigma',\sim')$ defined as follows:

 \vspace{1mm}

\begin{itemize}
\item[$\bullet$] $\Sigma'=\{(\sigma_1,x,(s,t))\mid \sigma_1(x)=s,\; \sigma_1\in\Sigma \}$, where $(\sigma_1,x,(s,t))$ is a new situation such that 
          $(\sigma_1,x,(s,t))(x)=t$ and $(\sigma_1,x,(s,t))(y)=\sigma_1(y)$.

\vspace{1mm}

       \item[$\bullet$] For any variable $z\in Var$ and situations $(\sigma_1, x, (s_1,t_1)), (\sigma_2, x, (s_2,t_2))\in\Sigma'$,   $(\sigma_1, x, (s_1,t_1))\sim'_z (\sigma_2, x, (s_2,t_2))$  if, and only if,

\vspace{1mm}
       
       \begin{itemize}
           \item[$(i)$] $\sigma_1\sim_z \sigma_2$, $(s_1,t_1)\approx_z (s_2,t_2)$.

\vspace{1mm}
           
\item[$(ii)$] If $\sigma_1(x)\in\mathbb{D}^k(\sigma_1(y))$, then $\sigma_1=\sigma_2$, otherwise $\sigma_2(x)\not\in\mathbb{D}^k(\sigma_2(y))$.
 \end{itemize} 
\end{itemize}
 \end{definition}

 \vspace{1mm}

In the definition above, the clause $(i)$ is the usual way to obtain the indistinguishability relations in product updates, and the clause $(ii)$ is an additional requirement that  ensures that the resulting models are still $k$-sight models.\footnote{As the case in \cite{johan-del-game}, one can also use two types of event models: one for movements, which can be obtained by removing the clause $(ii)$; and the other for ``inspection'', which can make the players know whether or not they are in the sight of each other. Then the desired outcome can be obtained by the product of the original epistemic model, the event model for movements and the event model for inspection (in this order). But here we mix the two types of event models for convenience.} 
Now we can formally re-analyze Example \ref{ex1} with the $\mathsf{DEL}$-approach specified above. Details are given in Figure \ref{fig:example-del}.  We add the game structure below to remind the readers.

  
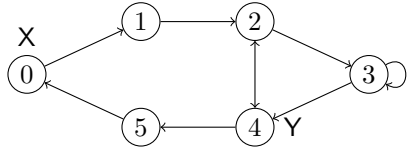
\vspace{1mm}
  
\begin{center}
\begin{tikzpicture}
\node(0)[circle,draw,inner sep=0pt,minimum size=5mm] at (0,0)[label=above:$\X$] {$0$};
\node(1)[circle,draw,inner sep=0pt,minimum size=5mm] at (1.5,.7) {$1$};
\node(2)[circle,draw,inner sep=0pt,minimum size=5mm] at (3,.7) {$2$};
\node(3) [circle,draw,inner sep=0pt,minimum size=5mm] at (4.5,0){$3$};
\node(4) [circle,draw,inner sep=0pt,minimum size=5mm] at (3,-.7)[label=right:$\Y$]{$4$};
\node(5) [circle,draw,inner sep=0pt,minimum size=5mm] at (1.5,-.7){$5$};
 
\draw[->] (0) to  (1);
\draw[->] (1) to  (2);
\draw[->](2) to  (3);
\draw[->](3) to  (4);
\draw[->](4) to  (5);
\draw[->](5) to  (0);
\draw[<->](2) to  (4);
\draw[->](3) to  [in=330, out=30,looseness=4] (3);
\end{tikzpicture}
\end{center}

  \vspace{1mm}

As mentioned in Figure \ref{fig:example-del}, each of the four layers is an epistemic model.
Notice that there is a striking analogy between the $\ELCR$-based analyses and the discussion here: in each layer with the class of situations $\Sigma$ and the actual situation $(s_1,s_2)$, the {\em part} $\Sigma|(s_1,s_2)$ corresponds to {\em complete}
situations in the $\ELCR$ setting. 
Although  the resulting sets of situations based on the $\ELCR$-updates  may increase (or decrease) in size, every  such situation  
makes sense for the analyses of the game.  We do not need to pay attention to the irrelevant situations.  Following this viewpoint we claim that $\ELCR$ is a more succinct proposal. 



\begin{figure} 
\begin{center}
 
     \begin{tikzpicture}[global scale = .9]
\node({04})[] at (0,0) {$\underline{(0,4)}$};
\node({03})[] at (3,0) {$(0,3)$};
\node({02})[] at (6,0) {$(0,2)$};
\node({14})[] at (-3,0) {$(1,4)$};

   \draw[-,dotted](04) to node [below] {$x$}  (03);
      \draw[-,dotted](03) to node [below] {$x$}  (02);
         \draw[-,dotted](04) to node [below] {$y$}  (14);

\node({a14})[] at (0,-1.5) {$\underline{(1,4)}$};
\node({a13})[] at (3,-1.5) {$(1,3)$};
\node({a12})[] at (6,-1.5) {$(1,2)$};
\node({a24})[] at (-3,-1.5) {$(2,4)$};

   \draw[->](04) to (a14);
 \draw[->](03) to (a13);
\draw[->] (02) to (a12);
  \draw[->] (14) to (a24);

     \draw[-,dotted](a14) to node [below] {$x$}  (a13);

\node({b15})[] at (1.5,-3) {$\underline{(1,5)}$};
\node({b12})[] at (0,-3) {$(1,2)$};
\node({b14})[] at (3,-3) {$(1,4)$};
\node({b13})[] at (6,-3) {$(1,3)$};
\node({b22})[] at (-3,-3) {$(2,2)$};
\node({b25})[] at (-1.5,-3) {$(2,5)$}; 

 \draw[->](a14) to (b15);
 \draw[->](a14) to (b12);
 \draw[->](a13) to (b13);
 \draw[->](a13) to (b14);
 \draw[->](a12) to (b13);
 \draw[->](a24) to (b22);
 \draw[->](a24) to (b25);

 \draw[-,dotted](b15) to node [below] {$x$}  (b14);
  \draw[-,dotted](b14) to node [below] {$x$}  (b13);

\node({c23})[] at (7.5,-4.5) {$(2,3)$};
\node({c24})[] at (4.5,-4.5) {$(2,4)$};
  \node({c25})[] at (3,-4.5) {$\underline{(2,5)}$};
\node({c22})[] at (1.5,-4.5) {$(2,2)$};

\node({c35})[] at (0,-4.5) {$(3,5)$}; 

\node({c45})[] at (-1.5,-4.5) {$(4,5)$};

\node({c32})[] at (-3,-4.5) {$(3,2)$}; 

\node({c42})[] at (-4.5,-4.5) {$(4,2)$};

 \draw[->](b22) to (c42);
  \draw[->](b22) to (c32);
   \draw[->](b25) to (c45);
  \draw[->](b25) to (c35);
   \draw[->](b12) to (c22);
    \draw[->](b15) to (c25);
     \draw[->](b14) to (c24);
      \draw[->](b13) to (c23);
\end{tikzpicture}
 
\end{center}
    \caption{Example analyses: A $\mathsf{DEL}$-approach. We use the dotted links labeled with $x$ and $y$ to represent the indistinguishability relations of the two players (we omit the self-loops and transitive links). Also, there are four layers  connected with solid arrows. Each of the  layers is an epistemic model, and for simplicity, we just draw the corresponding classes of situations and the indistinguishability relations, and highlight the actual situations with underlines: the first layer is the original epistemic model $M$, the second one is $M_1=(M,(0,4))\times E_{(0,1)}^x$, the third is  $M_2=(M_1,(1,4))\times E_{(4,5)}^y$, and the last layer is  $(M_2,(1,5))\times E_{(1,2)}^x$. Finally, the solid arrows indicate how the classes of situations in the latter stages come from the previous stages.}
    \label{fig:example-del}
\end{figure}
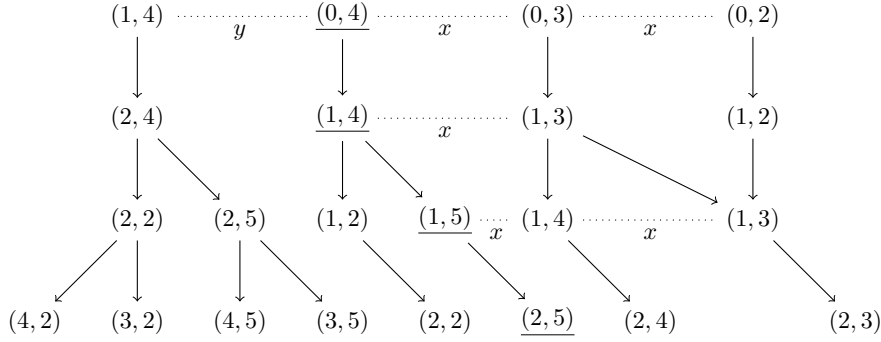

 \subsection{Other perspectives}\label{sec:related-works}
 
 As stated  earlier, this work is an extension of \cite{elcr}.
 The game and its various variants  have been well-studied in the field of computer science, with algorithmic and combinatorial perspectives (see e.g., \cite{cop-robber-book,cop-robber}).  Complementary to these approaches, subsequent studies from  the logical aspects have been  conducted in recent years for the perfect information version of the game, also termed as Hide and Seek in the literature. The first logical  proposal in this direction is  defined in \cite{graphgame}, which constitutes a broad program that promotes the
study of graph game design in tandem with matching new modal logics. Afterwards, this is studied extensively in \cite{lhs,LHS-journal,Chenqian2023}, involving its expressive power at the levels of models and frames, computational behavior, axiomatization and its connections with other relevant paradigms like {\em product logics with the diagonal constant} (e.g., \cite{kikot-non-finitely-axiomatizable, hybrid-K-delta-K}).  More recently, the original logic for the game has been extended with formulas from {\em hybrid logic} (e.g., \cite{hybrid-language}) in \cite{hlhs-axiomatization}, which are important to establish a complete Hilbert-style proof system for the resulting logic. In addition, based on the approach of substitutions \cite{Johan-fixedpoint2021}, \cite{substitution}  develops a logical proposal that is crucial to capture the winning positions of players in  a natural infinite setting.

 This work and the series of logical studies on Hide and Seek belong to the broader exploration on the interaction between graph game and logic, a trend starting from {\em sabotage games} \cite{original-sg,sabotage,sabotagelori} and its matching {\em sabotage modal logic}. Different from Cops and Robbers in which the game graphs are fixed, in each round of a sabotage game, a blocker tries to stop the other player from moving to a given goal region by removing a link from the graph. Many variants of sabotage games have been studied from a logical perspective \cite{Rohde-thesis,learning,dazhu-jolli}.\footnote{There are many further works on the sabotage modal logic. In \cite{hybrid-sml}, the logic is axiomatized in a broader setting with hybrid formulas, and \cite{lileibisimulation} offers an upper bound of the complexity to determine the notion of bisimulation for the logic. Also, in addition to the sabotage modal logic that contains an operator to delete links, there is a class of {\em relation-changing logics} that contain operators to swap and add links \cite{changeoperator4,changeoperator2,Penghao2023}.}  Also, \cite{linkdeletion} studies the graph games with  a particular policy of link deletion  that is performed  under certain conditions that can be expressed explicitly in a given language.  Distinct from link modifications, \cite{poisonlogic,poisonargumentation} develop logics to capture the so-called {\em poison games} \cite{poison-game}, in which a player can poison a node, to make it unavailable to the opponent. The logics of poison games are further explored in \cite{poison-penghao}, involving their axiomatization and computational behavior.  In addition, \cite{Declan} studies a dynamic logic of local fact changes that captures a class of graph games in which properties of vertices might be affected by other vertices. All these  studies on graph games are about the settings in which players have perfect information. Based on a $\DEL$-approach,
\cite{johan-del-game} analyzes some imperfect information versions, involving  the sabotage games and the game of Hide and Seek. 
For more on this topic, we refer to \cite{graphgame} for a broad research program, \cite{Dazhu-thesis} for extensive references to modal logics for graph games, and \cite{graph-game-book} for the latest developments of this area.

Finally, the design of the static fragment of our language is  
 inspired  
by the works on different sorts of dependencies: {\em the logic of epistemic dependency} 
\cite{knowing-value}  and {\em the logic of functional dependency}
\cite{lfd}.  The latter is about the dependencies between variables, and  contains formulas of the form $D_xy$, expressing that fixing the value of variable $x$ would determine the value of variable $y$. More  relevantly,  
the logical language of the former contains formulas $K_ax$ expressing that the agent $a$ knows the value of $x$, where $a$ is an index for agents and $x$ is a variable of the object language.  In contrast, our work introduces formulas of the form $K_xy$, meaning $x$ knows the location of $y$, with both $x$ and $y$ as variables. 
Moreover,
\cite{knowing-value}  explores the dynamic scenarios induced by  public announcement operators that are  central to 
$\mathsf{DEL}$, 
our focus is on agent movements, essential to the game of Cops and Robbers.
Combining these two approaches to dynamics -- public announcements and movement-based updates -- would certainly be an interesting direction for further study.

\section{Conclusion and future work}\label{sec:conclusion}

\noindent{\it Summary.}\; To study the game of Cops and Robbers  with uncertainty among players,  this paper provided a formal framework $\ELCR$ to capture 
players' reasoning about knowledge and actions. As illustrated, many validities of the logic characterized  natural assumptions on the game.  Axiomatization and decidability  of the static version $\ELCR^{-}$ and the full dynamic logic $\ELCR$ were explored.  From a broader perspective, we showed that the ideas underlying the design of $\ELCR$ can be easily adjusted to fit other important variants of the game and provided a formal connection between $\ELCR$ and the corresponding setting with simultaneous moves of players. 
In addition, we  formalized the ideas in the literature that advocate using the $\mathsf{DEL}$-approach  for studying the game.  In process, we set the stage for proving a succinctness result with respect to 
our update mechanism.   

 \vspace{1mm}

\noindent{\it Future work.}\; Several further directions have been identified in the article. In addition, there are other promising directions to explore. For example, 
there is extensive literature on the complexity of different versions of the Cops and Robbers game \cite{cop-robber-book}, and we intend to do the same for the one introduced in our work and its different variants, and in process, design 
efficient algorithms to construct winning strategies of players.
On the logic side, an immediate  extension would be to provide a complete Hilbert-style proof system for the logic designed for the simultaneous movements.  This could be obtained by adapting the recursion axioms provided in Section \ref{sec:axiomatization-whole-logic}.  Also, although Section \ref{sec:del} illustrates the succinctness of $\ELCR$, it remains to be examined how much more succinct it is than the $\DEL$-approach. 
Another important direction is to study other logical properties of $\ELCR$, including its expressiveness and frame correspondence.  Moreover, it is crucial to study the setting involving higher-order knowledge,  
especially the cases that players have limited abilities to reason about each other's knowledge 
\cite{diversity}.
Finally,  from the games perspective, graph games with imperfect information warrant a more detailed study, which could be facilitated by usage of logic tools developed in this work.  

 \vspace{3mm}

\noindent{\bf{Acknowledgements.}}\; 
This research was inspired by a question from Alexandru Baltag when the logic of the hide and seek game was presented at the January 2022 workshop `Exploring Baltag's Universe'. We thank Alexandru Baltag, Johan van Benthem, Davide Grossi, and Katsuhiko Sano for their valuable feedback, as well as the editors of the LORI special issue and the anonymous referees for their helpful comments.
Dazhu Li is supported by the National Social Science Foundation of China [22CZX063]. Sujata Ghosh acknowledges financial support from the Department of Science and Technology, Government of India (Ref. No. DST/CSRI/2018/202, CSRI). 
Fenrong Liu is supported by the Tsinghua University Initiative Scientific Research Program.

\bibliographystyle{plain}
\DeclareRobustCommand{\VAN}[3]{#3}
\bibliography{sn-bibliography}

\begin{appendices}

\section{Proof for Fact \ref{fact:soundness-static}}\label{sec:appendix-proof-for-soundness-static}

\begin{proof}


We will show the validity of axiom  $(\mathtt{Knowledge}$-$\mathtt{Ground})$, and prove that both $(\mathtt{K}$-$\mathtt{Additivity})$ and $(\mathtt{K}$-$\mathtt{Elimination})$  preserve validity. The others are left as exercises to the reader. Let $M=(\bD,\bI,\Sigma,\sim)$ be a $k$-sight model and $\sigma\in\Sigma$.

 \vspace{1mm}


 



 (1) First, we show $(\mathtt{Knowledge}$-$\mathtt{Ground})$ is valid. W.l.o.g., let $z:=x$ and $z':=y$. Assume that $M,\sigma\models K_xy=\T$, where $\T\subseteq Cons$. The proof for  the case that $\alpha$ does not contain any occurrence of $y$ is easy. Assume that the formula does contain $y$. 

 \vspace{1mm}

 (1.1) Assume that $M,\sigma\not\models \bigwedge_{c\in \T}\alpha[c/y]$. Then,  $M,\sigma\not\models \alpha[c/y]$ for some $c\in \T$. Since $c\in \T$, it follows from $M,\sigma\models K_xy=\T$ that there is some $\sigma'\in\Sigma$ s.t. $\sigma\sim_x\sigma'$ and $\sigma'(y)=\bI(c)$. By $M,\sigma\not\models \alpha[c/y]$,  $M,\sigma'\not\models \alpha$, which then gives us $M,\sigma\not\models K_x\alpha$. 

 \vspace{1mm}

 (1.2) Suppose that $M,\sigma\not\models K_x\alpha$. Then, there is a situation $\sigma'\in\Sigma$ such that $\sigma\sim_x\sigma'$ and $M,\sigma'\not\models\alpha$. By the axiom $(\mathtt{At}$-$\mathtt{Some}$-$\mathtt{Where})$, there is some $c'\in Cons$ such that $M,\sigma'\models c_1\equiv y$. Clearly, $c_1\in \T$. Since $M,\sigma'\not\models\alpha$, it holds that $M,\sigma'\not\models\alpha[c'/y]$. Recall that $\sigma\sim_x\sigma'$, $\sigma(x)=\sigma'(x)$ (due to the $0\le k$-sight ability), so  $M,\sigma\not\models\alpha[c'/y]$, which then shows that  $M,\sigma\not\models\bigwedge_{c\in \T}\alpha[c/y]$, as needed.

\vspace{1mm}

(2) We move to $(\mathtt{K}$-$\mathtt{Additivity})$. Assume that $\varphi\to\alpha$ is valid on any $k$-sight models based on $(\bD,\bI)$.  Let $M=(\bD,\bI,\Sigma,\sim)$ be a $k$-sight model and $\sigma,\sigma'\in\Sigma$ with $\sigma\sim_z\sigma'$ and $M,\sigma\models\varphi$. By Fact \ref{fact:trans-sym}, $M,\sigma'\models\varphi$. So,  $M,\sigma'\models\alpha$, which indicates  $M,\sigma\models K_z\alpha$. 

\vspace{1mm}

 (3) Finally, we prove that $(\mathtt{K}$-$\mathtt{Elimination})$ preserves the validity of formulas. Assume that $\varphi\to (K_z\alpha\to \beta)$ is valid on any $k$-sight models based on $(\bD,\bI)$. W.l.o.g., let $z:=x$ and $z':=y$.  Also, suppose that for some  $M=(\bD,\bI,\Sigma,\sim)$ and $\sigma\in\Sigma$, it holds that $M,\sigma\models\varphi\land \alpha\land\neg\beta$. Now, define $M'=(\bD,\bI,\Sigma,\sim')$  such that for any $\sigma_1,\sigma_2\in\Sigma$, $\sigma_1\sim'_x\sigma_2$ iff $\sigma_1=\sigma_2$, and for the other variable $y$, $\sigma_1\sim'_{y}\sigma_2$ iff  $\sigma_1\sim_{y}\sigma_2$. Now, by the form of $\varphi$, it is easy to see that $M',\sigma\models\varphi$. Also, by Fact \ref{fact:Boolean-locality}, $M',\sigma\models \alpha\land\neg\beta$. So, $M',\sigma\models K_x\alpha\land\neg\beta$, a contradiction.
\end{proof}

\section{Proof for Theorem \ref{theorem:soundness}}\label{sec:appendix-proof-for-soundness}


\begin{proof}
    Based on Fact \ref{fact:soundness-static}, it suffices to show the validity of the new axioms $(\mathtt{R1})$-$(\mathtt{R10})$. Let $M=(\bD,\bI,\Sigma,\sim)$ be a $k$-sight model and $\sigma_1\in\Sigma$. The validity of $(\mathtt{R2})$, $(\mathtt{R3})$, $(\mathtt{R4})$ and $(\mathtt{R7})$ is easy to see. In what follows, we consider   $(\mathtt{R1})$, $(\mathtt{R5})$, $(\mathtt{R8})$ and $(\mathtt{R10})$, and the key ideas of the proofs for $(\mathtt{R6})$ and $(\mathtt{R9})$ are similar to those of $(\mathtt{R10})$ and $(\mathtt{R5})$. 

\vspace{1mm}

(1) We consider $(\mathtt{R1})$. We have the following:
\begin{center}
    \begin{tabular}{rcl}
     $M,\sigma_1\models [z]\alpha$    &\quad   iff   &\quad    For all $\sigma_2\in \mathsf{R}^{z}(\sigma_1)$, $(\bD, \bI, \Sigma', \sim'), \sigma_2\models \alpha$  \\
&\quad   iff   &\quad   $M,\sigma_1\models \bigwedge_{\T\subseteq Cons}(Rz=\T\to \bigwedge_{c\in \T}\alpha[c/z])$
    \end{tabular}
\end{center}
The last equivalence holds by Fact \ref{fact:Boolean-connection}, as needed.

\vspace{1mm}





 (2) We consider  the formula $(\mathtt{R5})$. W.l.o.g., let $z:=x$ and $z':=y$. Assume that $\T,\T_1\subseteq Cons$ such that $M,\sigma_1\models K_xy=\T \land Rx=\T_1$.

 First, let  $M,\sigma_1\models [x]K_xy$. Then,  for all $\sigma_2\in\mathsf{R}^{x}(\sigma_1)$, $M',\sigma_2\models K_xy$. Suppose towards a contradiction that the following is not the case:
\begin{center}
   $  K_xy\lor   
 \bigwedge_{a\in \T_1}(\mathsf{D}^kay \lor  \bigwedge_{b\in \T}(b\not\equiv y\to \mathsf{D}^k ab))$
\end{center}
\noindent  Since  $M,\sigma_1\models\neg K_xy$, the $\T$ cannot be a set of constants that refer to the same vertex. Also, there are  constants $a\in \T_1$ and $b\in \T$ s.t. $M,\sigma_1\models \neg \mathsf{D}^kay \land b\not\equiv y\land  \neg \mathsf{D}^kab$. Let $\sigma:=\sigma_1[y:=b]$. Since $b\in \T$, we have $\sigma\in \Sigma$ and $\sigma\sim_x\sigma_1$, and so $\sigma\in \Sigma|\sigma_1$. Moreover, let $M,\sigma_1\models y\equiv c$ for some $c\in Cons$. Such a constant always exists (the axiom  $(\mathtt{At}$-$\mathtt{Some}$-$\mathtt{Where})$), and due to the fact $\sim_x$ is an equivalence relation, $c\in \T$. Assume that $x$ moves to $a$ and we write $\sigma_2$ for the new situation (i.e., $\sigma_2:=\sigma_1[x:=a]$). Consider $\sigma':=\sigma[x:=a]$. Obviously, $\sigma'\in \mathsf{R}^{x}(\Sigma|\sigma_1)$, and it follows from $M,\sigma_1\models \neg \mathsf{D}^kay$ that $\sigma'(y)\not\in\mathbb{D}^{k}(\sigma_2)$. So, $\sigma'\in\Sigma'$. Now, $M',\sigma\models \lr{K_x}y\equiv b\land y\equiv c\land b\not\equiv c$, a contradiction.

\vspace{1mm}

For the converse, assume that $M,\sigma_1\not\models [x]K_xy$, i.e., $M',\sigma_2\models \neg K_xy$ for some  $\sigma_2\in\mathsf{R}^{x}(\sigma_1)$. Then there is a $\sigma'\in\Sigma'$ s.t. $\sigma_2(x)=\sigma'(x)$ and $\sigma_2(y)\not=\sigma'(y)$. We consider the previous stage $\sigma:=\sigma'[x:=\sigma_3(x)]$ of $\sigma'$, where $\sigma_3\in \Sigma|\sigma_1$ (so $\sigma$ is $\sigma_3$).  Now,  at least one of $\sigma\sim_x\sigma_1$ and $\sigma\sim_y\sigma_1$ must be the case. Notice that $\sigma(y)=\sigma'(y)$ and $\sigma_1(y)=\sigma_2(y)$. Since $\sigma_2(y)\not=\sigma'(y)$, we have $\sigma\not\sim_y\sigma_1$ (due to the $0\le k$-sight ability), which then gives us 
 $\sigma\sim_x\sigma_1$ and so $M,\sigma_1\not\models K_xy$. Recall the assumption that $M,\sigma_1\models K_x y=\T\land Rx=\T_1$. Moreover, assume that $a$ is a constant with  $M',\sigma_2\models x\equiv a$. By $(\mathtt{At}$-$\mathtt{Some}$-$\mathtt{Where})$, such a constant always exists, and so $a\in A$. Also, let $b$ be a constant such that $M,\sigma\models y\equiv b$. Then, we can see that $b\in \T$. Now, we have $M,\sigma_1\models \neg \mathsf{D}^kay \land b\not\equiv y\land \neg \mathsf{D}^kab$, as needed.
 
 \vspace{1mm}

(3) We  move to $(\mathtt{R8})$. Note that for any $\sigma_2\in\mathsf{R}^{z}(\sigma_1)$, $\sigma_1(z')=\sigma_2(z')$, where $z'$ is the variable distinct from $z$. By Fact \ref{fact:Boolean-locality},   $M,\sigma_1\models \alpha$ iff $M',\sigma_2\models \alpha$. Since $\alpha$ does not contain $z'$, $M',\sigma_2\models \alpha$ iff $M',\sigma_2\models K_z\alpha$. Now,   $M,\sigma_1\models[z]K_{z'}\alpha$ iff $M,\sigma_1\models \alpha$.

\vspace{1mm}

(4) Let us  consider $(\mathtt{R9})$. Let $z:=x$ and $z':=y$. Also, let $\T,\T_1\subseteq Cons$ such that  $M,\sigma_1\models Rx=\T \land K_xy=\T_1$.

\vspace{1mm}

(4.1) For the direction from left to right,  assume that $M,\sigma_1\models [x]K_x\alpha(y)$, i.e., for any  $\sigma_2\in\mathsf{R}^{x}(\sigma_1)$, $M',\sigma_2\models K_x\alpha(y)$. Suppose  that the following is not the case:
\begin{align*}
 & (K_xy\land \bigwedge_{c\in \T}\alpha[c/x])\lor  \\
&(\neg K_xy\land \bigwedge_{t\in \T}((\mathsf{D}^kty\land \alpha[t/x])\lor  
  (\neg \mathsf{D}^kty\land \bigwedge_{c\in \T_1}(\neg \mathsf{D}^ktc\to  \alpha[t/x][c/y]))))   
\end{align*}

 (4.1.1) We first assume that $M,\sigma_1\models K_xy$. Then, by the first disjunct, there is some $c\in \T$ such that $M,\sigma_1\not\models \alpha[c/x]$.  But by Fact \ref{fact:Boolean-connection}, we arrive at a contradiction.
 
 \vspace{1mm}

 (4.1.2) Assume that $M,\sigma_1\not\models K_xy$. Then, there is a constant $c\in \T$ such that exactly one of the following is the case:
\begin{itemize}
    \item[$(a)$] $\mathsf{D}^kcy$ and $\neg \alpha[c/x]$
    \item[$(b)$] $\neg \mathsf{D}^kcy$, and there is some $c_1\in \T_1$ such that $\neg \mathsf{D}^kcc_1$ and $\neg\alpha[c/x][c_1/y]$.
\end{itemize}
We fix an assignment $\sigma':=\sigma_1[x:=\bI(c)]$.
When $(a)$ is the case,  by definition,  we have $\Sigma'=\{\sigma'\}$ in the pointed model $M',\sigma'$. Then, it follows from $\neg \alpha[c/x]$ that $M',\sigma'\not\models \alpha$, which  gives us $M',\sigma'\not\models K_x\alpha$, a contradiction. Let us assume that $(b)$ is the case. Since there is some $c_1\in \T_1$ with $\neg \mathsf{D}^kcc_1$, for the assignment $\sigma_3:=\sigma_1[y:=\bI(c_1)]$ and $\sigma'_3:=\sigma_3[x:=\bI(c)]$, we have $\sigma_3\sim_x\sigma_1$ and $\sigma'_3 \in\Sigma'$. It holds that $\sigma'\sim'_x\sigma'_3$. Also, $M',\sigma'_3\models \neg\alpha[c/x][c_1/y]$, and so $M',\sigma'\not\models K_x\alpha[c/x]$, which  gives us $M',\sigma'\not\models K_x\alpha$.

\vspace{1mm}

(4.2) We consider the other direction. Suppose that $M,\sigma_1\not\models [x]K_x\alpha(y)$, i.e., there is some  $\sigma_2\in\mathsf{R}^{x}(\sigma_1)$ with $M',\sigma_2\not\models K_x\alpha(y)$. We discuss different cases.

\vspace{1mm}

(4.2.1) We assume that $M,\sigma_1\models K_xy$. Then,  $M',\sigma_2\models K_xy$ (Fact \ref{fact:movement-theother}). Given this, it follows from $M',\sigma_2\not\models K_x\alpha(y)$ that $M',\sigma_2\not\models \alpha(y)$. By Fact \ref{fact:Boolean-connection}, $M,\sigma_1\not\models \bigwedge_{c\in \T}\alpha[c/x]$.

\vspace{1mm}

(4.2.2) Assume $M,\sigma_1\not\models K_xy$. Then, there is an assignment $\sigma\in\Sigma$ s.t. $\sigma_1\sim_x\sigma$ and $\sigma_1(y)\not=\sigma(y)$. Let $\bI(c_1)=\sigma(y)$  and $\bI(c)=\sigma_2(x)$. Clearly, $c_1\in \T_1$ and $c\in \T$. There are two different cases. We first suppose that $\sigma_2(x)\in\mathbb{D}^{k}(\bI(c_1))$. Then, the $\Sigma'$ in $M'$ is  $\{\sigma_2\}$, and the fact $M',\sigma_2\not\models K_x\alpha(y)$ indicates $M',\sigma_2\not\models \alpha$, and so,  $M',\sigma_2\not\models \alpha[c/x]$.  Next,   assume $\sigma_2(x)\not\in\mathbb{D}^{k}(\bI(c_1))$. Since $M',\sigma_2\not\models K_x\alpha(y)$, there is $\sigma'_3\in\Sigma'$ s.t. $\sigma_2\sim_x\sigma'_3$ and $M',\sigma'_3\not\models \alpha(y)$. Let $c_2,c_3\in Cons$ s.t. $M',\sigma'_3 \models c_2\equiv y\land c_3\equiv x$. Clearly, $c_2\in \T_1$ and $c_3\in \T$. It follows from $\sigma'_3\in\Sigma'$ and $\sigma_2(x)\not\in\mathbb{D}^{k}(\bI(c_1))$ that $\sigma'_3(x)\not\in\mathbb{D}^{k}(\sigma'_3(y))$, i.e., $M',\sigma'_3\models \neg \mathsf{D}^{k}c_3c_2$. Since $M',\sigma'_3\not\models \alpha(y)$,  $M',\sigma'_3\not\models \alpha[c_3/x][c_2/y]$. 

\vspace{1mm}

Putting (4.2.1) and (4.2.2) together, we can get the negation of the right part of the equivalence $(\mathtt{R9})$, as needed. 
\end{proof}

\end{appendices}

\end{document}